\documentclass[a4paper,english,final]{lipics-v2018}

\pdfoutput=1

\usepackage[inline]{enumitem}
\usepackage{subcaption}
\usepackage{tasks}
\usepackage{hyperref}
\usepackage{microtype}
\usepackage[utf8]{inputenc}
\usepackage{tikz}
\usepackage{xspace}
\usepackage{amsmath, amssymb, amsthm}
\usepackage{etoolbox}

\usetikzlibrary{arrows,decorations.pathreplacing}

\DeclareMathOperator{\poly}{\operatorname{poly}}

\newcommand{\NP}{$\mathsf{NP}$\xspace}
\newcommand{\XP}{$\mathsf{XP}$\xspace}
\newcommand{\FPT}{$\mathsf{FPT}$\xspace}
\newcommand{\W}[1]{$\mathsf{W[#1]}$\xspace}
\newcommand{\Wh}[1]{$\mathsf{W[#1]}$-hard\xspace}
\newcommand{\II}{\mathcal I}

\newcommand{\Oh}{\mathcal O}
\newcommand{\OhOp}[1]{{\mathcal O}\mathopen{}\mathclose\bgroup\left( #1 \aftergroup\egroup\right)}
\newcommand{\OhStarOp}[1]{{\mathcal O}^*\mathopen{}\mathclose\bgroup\left( #1 \aftergroup\egroup\right)}

\def \r {\rangle}

\newcommand{\RR}{\mathbb{R}}
\newcommand{\QQ}{\mathbb{Q}}
\newcommand{\ZZ}{\mathbb{Z}}
\newcommand{\NN}{\mathbb{N}}
\newcommand{\G}{\mathcal{G}}
\DeclareMathOperator{\nd}{{\rm nd}}
\DeclareMathOperator{\vc}{{\rm vc}}
\DeclareMathOperator{\td}{{\rm td}}
\DeclareMathOperator{\tw}{{\rm tw}}
\DeclareMathOperator{\cw}{{\rm cw}}
\DeclareMathOperator{\mw}{{\rm mw}}

\DeclareMathOperator{\Succ}{{\rm succ}}
\DeclareMathOperator{\sign}{{\rm sign}}
\newcommand{\transpose}{\intercal}

\title{Integer Programming in Parameterized Complexity: Three Miniatures
}
\titlerunning{Integer Programming in Parameterized Complexity: Three Miniatures} 

\author{Tomáš Gavenčiak}{Department of Applied Mathematics, Charles University, Prague, Czech Republic}{gavento@kam.mff.cuni.cz}{}{}

\author{Dušan Knop}{Department of Informatics, University of Bergen, Bergen, Norway}{dusan.knop@gmail.com}{}{Author supported by the project NFR MULTIVAL.}

\author{Martin Koutecký}{Technion - Israel Institute of Technology, Haifa, Israel \emph{and} \\Computer Science Institute, Charles University, Prague, Czech Republic}{koutecky@technion.ac.il}{}{Author supported by Technion postdoctoral fellowship.}

\authorrunning{T. Gavenčiak, D. Knop, M. Koutecký}  

\Copyright{Tomáš Gavenčiak, Dušan Knop, Martin Koutecký}

\subjclass{
Theory of computation $\rightarrow$ Parameterized complexity and exact algorithms;
Theory of computation $\rightarrow$ Graph algorithms analysis
}

\category{}

\relatedversion{\url{https://arxiv.org/abs/1711.02032} (full version)}

\supplement{}

\funding{}

\acknowledgements{}

\EventEditors{Christophe Paul and Micha\l{} Pilipczuk}
\EventNoEds{2}
\EventLongTitle{13th International Symposium on Parameterized and Exact Computation (IPEC 2018)}
\EventShortTitle{IPEC 2018}
\EventAcronym{IPEC}
\EventYear{2018}
\EventDate{August 20--24, 2018}
\EventLocation{Helsinki, Finland}
\EventLogo{}
\SeriesVolume{115}
\ArticleNo{21}
\nolinenumbers 
\hideLIPIcs  

\theoremstyle{definition}
\newtheorem{model}[theorem]{Model}

\BeforeBeginEnvironment{model}{
\noindent\makebox[\linewidth]{\rule{\textwidth}{0.4pt}}
\vspace*{-22pt}
}
\AfterEndEnvironment{model}{
\vspace*{-14pt}
\noindent\makebox[\linewidth]{\rule{\textwidth}{0.4pt}}
}

\newtheorem{observation}{Observation}
\newtheorem{claim}[observation]{Claim}
\newcommand{\cqed}{\renewcommand{\qed}{\hfill $\vartriangleleft$}}

\usepackage{tabularx}

\newcommand{\Trule}{\rule{0pt}{3ex}}
\newcommand{\Brule}{\rule[-2ex]{0pt}{0pt}}

\newcommand{\prob}[3]{
\begin{center}
\begin{tabularx}{\textwidth}{|l X|}
	\hline
	\multicolumn{2}{|c|}{#1\Trule} \\
	\textbf{Input:\ }&{#2}\\
	\textbf{Task:\ }&{#3\Brule}\\
	\hline
\end{tabularx}
\end{center}
}

\def\ve#1{\mathchoice{\mbox{\boldmath$\displaystyle\bf#1$}}
    {\mbox{\boldmath$\textstyle\bf#1$}}
    {\mbox{\boldmath$\scriptstyle\bf#1$}}
    {\mbox{\boldmath$\scriptscriptstyle\bf#1$}}}

\newcommand\veb{{\ve b}}

\newcommand\ved{{\ve d}}
\newcommand\vece{{\ve e}}
\newcommand\vef{{\ve f}}
\newcommand\veg{{\ve g}}
\newcommand\veh{{\ve h}}
\newcommand\vel{{\ve l}}

\newcommand\veu{{\ve u}}

\newcommand\vew{{\ve w}}
\newcommand\vex{{\ve x}}
\newcommand\vey{{\ve y}}
\newcommand\vez{{\ve z}}

\newcommand{\N}{\mathbb{N}}

\usepackage{etoolbox}

\newcommand{\sv}[1]{}
\newcommand{\lv}[1]{#1}

\newcommand{\appendixText}{}

\newcommand{\toappendix}[1]{#1}

\newcommand{\appmark}{$\star$}

\newcommand{\tool}[1]{\vspace{0pt}\noindent\textbf{#1}}
\newcommand{\comment}[1]{\textcolor{black!60}{#1}}

\settasks{
before-skip = -.5em,
after-item-skip = 0.3333em
}

\NewTasks[label-align=left,after-item-skip=-0em,label-offset=0em]{ilps}(2)

\BeforeBeginEnvironment{ilps}{\begin{changemargin}{-2.5em}{0em}}
\AfterEndEnvironment{ilps}{\end{changemargin}}
\AtBeginEnvironment{enumerate}{\vspace*{-.6em}}

\begin{document}

\maketitle

\begin{abstract}
Powerful results from the theory of integer programming have recently led to substantial advances in parameterized complexity.
However, our perception is that, except for Lenstra's algorithm for solving integer linear programming in fixed dimension, there is still little understanding in the parameterized complexity community of the strengths and limitations of the available tools.
This is understandeable: it is often difficult to infer exact runtimes or even the distinction between \FPT and \XP algorithms, and some knowledge is simply unwritten folklore in a different community.
We wish to make a step in remedying this situation.

To that end, we first provide an easy to navigate quick reference guide of integer programming algorithms from the perspective of parameterized complexity.
Then, we show their applications in three case studies, obtaining \FPT algorithms with runtime $f(k) \poly(n)$. We focus on:
\begin{itemize}
\item \emph{Modeling}: since the algorithmic results follow by applying existing algorithms to new models, we shift the focus from the complexity result to the modeling result, highlighting common patterns and tricks which are used.
\item \emph{Optimality program:} after giving an \FPT algorithm, we are interested in reducing the dependence on the parameter; we show which algorithms and tricks are often useful for speed-ups.
\item \emph{Minding the $\poly(n)$:} reducing $f(k)$ often has the unintended consequence of increasing $\poly(n)$; so we highlight the common trade-offs and show how to get the best of both worlds.
\end{itemize}

Specifically, we consider graphs of bounded neighborhood diversity which are in a sense the simplest of dense graphs, and we show several \FPT algorithms for \textsc{Capacitated Dominating Set}, \textsc{Sum Coloring}, and \textsc{Max-$q$-Cut} by modeling them as convex programs in fixed dimension, $n$-fold integer programs, bounded dual treewidth programs, and indefinite quadratic programs in fixed dimension.
\end{abstract}

\section{Introduction}
 Our focus is on modeling various problems as \textsc{integer programming} (IP), and then obtaining \FPT algorithms by applying known algorithms for IP.
IP is the problem
\begin{equation}
\min \{f(\vex) \mid \vex \in S \cap \ZZ^n, S \subseteq \RR^n \text{ is convex}\} \enspace . \tag{IP} \label{IP}
\end{equation}
We give special attention to two restrictions of IP.
First, when $S$ is a polyhedron, we get
\begin{equation}
\min \{f(\vex)\mid A \vex \leq \veb, \, \vex \in \ZZ^n\}  \tag{LinIP},\label{LinIP}
\end{equation}
where $A \in \ZZ^{m \times n}$ and $\veb \in \ZZ^m$; we call this problem~\emph{linearly-constrained IP}, or \textsc{LinIP}.
Further restricting $f$ to be a linear function gives \textsc{Integer Linear Programming} (ILP):
\begin{equation}
\min \{\vew \vex \mid A \vex \leq \veb, \, \vex \in \ZZ^n\}  \tag{ILP},\label{ILP}
\end{equation}
where $\vew \in \ZZ^n$.
The function $f\colon \ZZ^n\to\ZZ$ is called the \emph{objective function}, $S$ is the \emph{feasible set} (defined by \emph{constraints} or various \emph{oracles}), and $\vex$ is a vector of \emph{(decision) variables}.
By $\langle \cdot \rangle$ we denote the binary encoding length of numbers, vectors and matrices.

In 1983 Lenstra showed that ILP is polynomial in fixed dimension and solvable in time $n^{\Oh(n)} \langle A, \veb, \vew \rangle$ (including later improvements~\cite{FrankTardos87,Kannan87,Lenstra83}).
Two decades later this algorithm's potential for applications in parameterized complexity was recognized, e.g. by Niedermeier~\cite{Niedermeier04}:
\begin{quote}
  \emph{[...] It remains to investigate further examples besides
    \textsc{Closest String} where the described ILP approach turns out
    to be applicable. More generally, it would be interesting to
    discover more connections between fixed-parameter algorithms and
    (integer) linear programming.}
\end{quote}
This call has been answered in the following years, for example in the context of graph algorithms~\cite{FellowsLMRS08,FialaGKKK:2017,Ganian12,Lampis12}, scheduling~\cite{HermelinKSTW:2015,JansenS:2010,KnopK:2016,MnichW:2015} or computational social choice~\cite{BredereckFNST:2015}.

In the meantime, many other powerful algorithms for IP have been devised; however it seemed unclear exactly \emph{how} could these tools be used, as Lokshtanov states in his PhD thesis~\cite{Lokshtanov09}, referring to \FPT algorithms for convex IP in fixed dimension:
\begin{quote}
  \emph{It would be interesting to see if these even more general
    results can be useful for showing problems fixed parameter
    tractable.}
\end{quote}
Similarly, Downey and Fellows~\cite{DowneyF13} highlight the \FPT algorithm for so called $n$-fold IP:
\begin{quote}
  \emph{Conceivably, [\textsc{Minimum Linear Arrangement}] might also
    be approached by the recent (and deep) FPT results of Hemmecke,
    Onn and Romanchuk~\cite{HemmeckeOR13} concerning nonlinear
    optimization.}
\end{quote}
Interestingly, \textsc{Minimum Linear Arrangement} was shown to be \FPT by yet another new algorithm for IP due to Lokshtanov~\cite{Lokshtanov:2015}.

In the last 3 years we have seen a surge of interest in, and an increased understanding of, these IP techniques beyond Lenstra's algorithm, allowing significant advances in fields such as parameterized scheduling~\cite{ChenM:2018,HermelinKSTW:2015,JansenKMR:2018,KnopK:2016,MnichW:2015}, computational social choice~\cite{KnopKM:2017b,KnopKM:2017,KnopKM:2018}, multichoice optimization~\cite{GajarskyHKO:2017}, and stringology~\cite{KnopKM:2017b}.
This has increased our understanding of the strengths and limitations of each tool as well as the modeling patterns and tricks which are typically applicable and used.

\subsection{Our Results}
We start by giving a quick overview of existing techniques in Section~\ref{sec:toolbox}, which we hope to be an accessible reference guide for parameterized complexity researchers.
Then, we resolve the parameterized complexity of three problems when parameterized by the neighborhood diversity of a graph (we defer the definitions to the relevant sections).
However, since our complexity results follow by applying an appropriate algorithm for IP, we also highlight our modeling results.
Moreover, in the spirit of the optimality program (introduced by Marx~\cite{Marx:2012}), we are not content with obtaining \emph{some} \FPT algorithm, but we attempt to decrease the dependence on the parameter $k$ as much as possible.
This sometimes has the unintended consequence of increasing the polynomial dependence on the graph size $|G|$.
We note this and, by combining several ideas, get the ``best of both worlds''.
Driving down the $\poly(|G|)$ factor is in the spirit of ``minding the $\poly(n)$'' of Lokshtanov et al.~\cite{LokshtanovRS:2016}.

We denote by $|G|$ the number of vertices of the graph $G$ and by $k$ its neighborhood diversity; graphs of neighborhood diversity $k$ have a succinct representation (constructible in linear time) with $\Oh(k^2 \log |G|)$ bits and we assume to have such a representation on input.

\begin{theorem} \label{thm:cds}
\textsc{Capacitated Dominating Set}
\begin{enumerate}[label={\alph*)}]
\item\label{thm:cds:convex} Has a convex IP model in $\Oh(k^2)$ variables and can be solved in time and space $k^{\Oh(k^2)} \log |G|$.
\item\label{thm:cds:ilp} Has an ILP model in $\Oh(k^2)$ variables and $\Oh(|G|)$ constraints, and can be solved in time $k^{\Oh(k^2)} \poly(|G|)$ and space $\poly(k,|G|)$.
\item\label{thm:cds:speed} Can be solved in time $k^{\Oh(k)} \poly(|G|)$ using model~\ref{thm:cds:convex} and a proximity argument.
\item\label{thm:cds:approx} Has a polynomial $OPT+k^2$ approximation algorithm by rounding a relaxation of~\ref{thm:cds:convex}.
\end{enumerate}
\end{theorem}

\begin{theorem} \label{thm:sumcol}
\textsc{Sum Coloring}
\begin{enumerate}[label={\alph*)}]
\item Has an $n$-fold IP model in $\Oh(k|G|)$ variables and $\Oh(k^2|G|)$ constraints, and can be solved in time $k^{\Oh(k^3)} |G|^2 \log^2 |G|$. \label{thm:sumcol:nfold}
\item Has a \textsc{LinIP} model in $\Oh(2^k)$ variables and $k$ constraints with a non-separable convex objective, and can be solved in time $2^{2^{k^{\Oh(1)}}} \log |G|$. \label{thm:sumcol:convex}
\item Has a \textsc{LinIP} model in $\Oh(2^k)$ variables and $\Oh(2^k)$ constraints whose constraint matrix has dual treewidth $k+2$ and whose objective is separable convex, and can be solved in time $k^{\Oh(k^2)} \log |G|$. \label{thm:sumcol:graver}
\end{enumerate}
\end{theorem}

\begin{theorem} \label{thm:maxcut}
\textsc{Max-$q$-Cut} has a \textsc{LinIP} model with an indefinite quadratic objective and can be solved in time $g(q,k)\log |G|$ for some computable function $g$.
\end{theorem}

\subsection{Related Work}
Graphs of neighborhood diversity constitute an important stepping stone in the design of algorithms for dense graphs, because they are in a sense the simplest of dense graphs~\cite{AlvesDFKSS16,AravindKKL:2017,BonnetS:2016,FialaGKKK:2017,Ganian12,Gargano:2015,MasarikT:16}.
Studying the complexity of \textsc{Capacitated Dominating Set} on graphs of bounded neighborhood diversity is especially interesting because it was shown to be \Wh{1} parameterized by treewidth by Dom et al.~\cite{DomLSV:2008}.
\textsc{Sum Coloring} was shown to be \FPT parameterized by treewidth~\cite{Jansen:1997}; its complexity parameterized by clique-width is open as far as we know.
\textsc{Max-$q$-Cut} is \FPT parameterized by $q$ and treewidth (by reduction to CSP), but \Wh{1} parameterized by clique-width~\cite{FominGLS14}.

\subsection{Preliminaries} \label{sec:prelims}
\sv{\toappendix{\section{Additional Material to Section~\ref{sec:prelims}}}}
For positive integers $m,n$ with $m \leq n$ we set $[m,n] = \{m,\ldots, n\}$ and $[n] = [1,n]$.
We write vectors in boldface (e.g., $\vex, \vey$) and their entries in normal font (e.g., the $i$-th entry of~$\vex$ is~$x_i$).
\lv{For an integer $a \in \ZZ$, we denote by $\langle a \rangle = 1 + \log_2 a$ the binary encoding length of $a$; we extend this notation to vectors, matrices and tuples of these objects.
For example, $\langle A, \veb \rangle = \langle A \rangle + \langle \veb \rangle$, and $\langle A \rangle = \sum_{i,j} \langle a_{ij} \rangle$.}
For a graph~$G$ we denote by $V(G)$ its set of vertices, by $E(G)$ the set of its edges, and by $N_G(v) = \{u \in V(G) \mid uv \in E(G) \}$ the (open) neighborhood of a vertex $v \in V(G)$.
For a matrix $A$ we define
\begin{itemize}
\item the \emph{primal graph} $G_P(A)$, which has a vertex for each column and two vertices are connected if there exists a row such that both columns are non-zero, and,
\item the \emph{dual graph} $G_D(A) = G_P(A^{\transpose})$, which is the above with rows and columns swapped.
\end{itemize}
We call the treedepth and treewidth of $G_P(A)$ the \emph{primal treedepth $\td_P(A)$} and \emph{primal treewidth $\tw_P(A)$}, and analogously for the \emph{dual treedepth $\td_D(A)$} and \emph{dual treewidth $\tw_D(A)$}.

We define a partial order $\sqsubseteq$ on $\RR^n$ as follows: for $\vex,\vey\in\RR^n$ we write $\vex\sqsubseteq \vey$ and say that $\vex$ is
{\em conformal} to $\vey$ if $x_iy_i\geq 0$ (that is, $\vex$ and $\vey$ lie in the same orthant) and $|x_i|\leq |y_i|$ for all $i \in [n]$.
It is well known that every subset of $\ZZ^n$ has finitely many
$\sqsubseteq$-minimal elements.

\begin{definition}[Graver basis]\label{def:graver}
The {\em Graver basis} of $A \in \ZZ^{m \times n}$ is the finite set $\G(A)\subset\ZZ^n$ of $\sqsubseteq$-minimal elements in $\{\vex\in\ZZ^n\,\mid\, A\vex=0,\ \vex\neq \mathbf{0}\}$.
\end{definition}

\subparagraph{Neighborhood Diversity.}
Two vertices $u,v$ are called \emph{twins} if $N(u)\setminus\{v\} = N(v)\setminus\{u\}$.
The \emph{twin equivalence} is the relation on vertices of a graph where two vertices are equivalent if and only if they are twins.
\begin{definition}[Lampis~\cite{Lampis12}]
The \emph{neighborhood diversity} of a graph $G$, denoted by $\nd(G)$, is the minimum number $k$ of classes (called \textit{types}) of the twin equivalence of $G$.
\end{definition}

\toappendix{
\lv{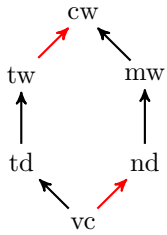
\begin{figure}[bt]}
\sv{\begin{figure}[!h]}
  \begin{minipage}[c]{0.2\textwidth}
    \begin{tikzpicture}
  [align=center, node distance=1.15cm]
  \tikzstyle{parameter}=[inner sep=3pt]

  \node[parameter](cw) {cw};
  \node[parameter, below left of=cw](tw) {tw};
  \node[parameter, below of=tw](td) {td};
  \node[parameter, below right of=cw](mw) {mw};
  \node[parameter, below of=mw](nd) {nd};
  \node[parameter, below left of=nd](vc) {vc};

  \draw[->,thick,red,>=stealth'] (tw) -- (cw);
  \draw[->,thick,>=stealth', red] (vc) -- (nd);
  \draw[->,thick,>=stealth'] (mw) -- (cw);
  \draw[->,thick,>=stealth'] (nd) -- (mw);
  \draw[->,thick,>=stealth'] (vc) -- (td);
  \draw[->,thick,>=stealth'] (td) -- (tw);
  
\end{tikzpicture}
  \end{minipage}
  \begin{minipage}[c]{0.78\textwidth}
    \caption[A map of parameters.]{\label{fig:parameterMap}A map of assumed parameters: $\vc$ is the vertex cover number, $\tw$ is treewidth, $\cw$ is clique-width, $\nd$ is  neighborhood diversity, and $\mw$ is modular-width.
Black arrow stands for linear upper bounds, while a red arrow stands for exponential upper bounds.
Note that treewidth and neighborhood diversity are incomparable because,
     \begin{minipage}{\linewidth}
     \vspace{-.3cm}
     \begin{align*}
     \tw(K_n) = n-1 && \nd(K_n) = 1 \\
     \tw(P_n) = 1  && \nd(P_n) = n,
     \end{align*}
     \end{minipage}
    where $K_n$ and $P_n$ are the complete graph and path on $n$ vertices, respectively.}
  \end{minipage}
\end{figure}
}


We denote by $V_i$ the classes of twin equivalence on $G$ for $i \in [k]$.
A graph $G$ with $\nd(G) = k$ can be described in a compressed way using only $\OhOp{\log |G| \cdot k^2}$ space by its type graph, which is computable in linear time~\cite{Lampis12}:
\begin{definition}
The \emph{type graph} $T(G)$ of a graph $G$ is a graph on $k=\nd(G)$ vertices $[k]$, where each $i$ is assigned weight $|V_i|$, and where ${i,j}$ is an edge or a loop in $T(G)$ if and only if two distinct vertices of $V_i$ and $V_j$ are adjacent.
\end{definition}

\subparagraph{Modeling.}
Loosely speaking, by \emph{modeling} an optimization problem $\Pi$ as a different problem $\Lambda$ we mean encoding the features of $\Pi$ by the features of $\Lambda$, such that the optima of $\Lambda$ encode \emph{at least some} optima of $\Pi$.
Modeling differs from reduction by highlighting which features of $\Pi$ are captured by which features of $\Lambda$.

In particular, when modeling $\Pi$ as an integer program, the same feature of $\Pi$ can often be encoded in several ways by the variables, constraints or the objective.
For example, an objective of $\Pi$ may be encoded as a convex objective of the IP, or as a linear objective which is lower bounded by a convex constraint; similarly a constraint of $\Pi$ may be modeled as a linear constraint of IP or as minimizing a penalty objective function expressing how much is the constraint violated.
Such choices greatly influence which algorithms are applicable to solve the resulting model.
Specifically, in our models we focus on the parameters \#variables (dimension), \#constraints, the largest coefficient in the constraints $\|A\|_\infty$ (abusing the notation slightly when the constraints are not linear), the largest right hand side $\|\veb\|_\infty$, the largest \emph{domain} $\|\veu-\vel\|_\infty$, and the largest coefficient of the objective function $\|\vew\|_\infty$ (linear objectives), $\|Q\|_\infty$ (quadratic objectives) or $f_{\max} = \max_{\vex: \vel \leq \vex \leq \veu} |f(\vex)|$ (in general), and noting other relevant features.

\subparagraph*{Solution structure.} We concur with Downey and Fellows that \FPT and structure are essentially one~\cite{DowneyF13}.
Here, it typically means restricting our attention to certain structured solutions and showing that nevertheless such structured solutions contain optima of the problem at hand.
We always discuss these structural properties before formulating a model.

\section{Integer Programming Toolbox}\label{sec:toolbox}
We give a list of the most relevant algorithms solving IP, highlighting their fastest known runtimes (marked \textbf{$\top$}), typical use cases and strengths (\textbf{$+$}), limitations (\textbf{$-$}), and a list of references to the algorithms (\textbf{$\heartsuit$}) and their most illustrative applications (\textbf{$\triangleright$}), both in chronological order.
\lv{We are deliberately terse here and defer a more nuanced discussion to Appendix~\ref{sec:cip_pc}.}

\subsection{Small Dimension}
The following tools generally rely on results from discrete geometry.
\lv{
Consider for example Lenstra's theorem: it can be (hugely) simplified as follows.
Let $S = \{\vex \mid A\vex \leq \veb\} \subseteq \RR^n$; then we can decide whether $S \cap \ZZ^n$ by the following recursive argument:
\begin{enumerate}
\item Either the volume of $S$ is too large not to contain an integer point by Minkowski's first theorem,
\item or the volume of $S$ is small and $S$ must be ``flat'' in some direction by the flatness theorem; then, we can cut $S$ up into few lower-dimensional slices and recurse into these.
\end{enumerate}
Being able to optimize an objective then follows from testing feasibility by binary search.
}

\medskip
\tool{ILP in small dimension.} Problem~\eqref{ILP} with small $n$.
\begin{ilps}
\task*[\textbf{$\top$}] $n^{2.5n} \langle A, \veb, \vew \r$~\cite{Kannan87,FrankTardos87}
\task*[\textbf{$+$}] Can use large coefficients, which allows encoding logical connectives using Big-$M$ coefficients~\cite{ilptricks}. Runs in polynomial space. Most people familiar with ILP.
\task*[\textbf{$-$}] Small dimension can be an obstacle in modeling polynomially many ``types'' of objects~\cite[Challenge \#2]{BredereckCFGNW:2014}.
Models often use exponentially many variables in the parameter, leading to double-exponential runtimes (applies to all small dimension techniques below).
Encoding a convex objective or constraint requires many constraints (cf. Model~\ref{mod:cds:linear}). Big-$M$ coefficients are impractical.
\task*[\textbf{$\heartsuit$}]Lenstra~\cite{Lenstra83}, Kannan~\cite{Kannan87}, Frank and Tardos~\cite{FrankTardos87}
\task*[\textbf{$\triangleright$}] Niedermeier (\textsc{Closest String})~\cite{Niedermeier04} Fellows et al. (graph layout problems)~\cite{FellowsLMRS08} Jansen and Solis-Oba (scheduling; MILP column generation technique)~\cite{JansenS:2010}, Fiala et al. (graph coloring)~\cite{FialaGKKK:2017}, Faliszewski et al. (computational social choice; big-$M$ coefficients to express logical connectives)~\cite{FaliszewskiGKT:2018}.
\end{ilps}

\tool{Convex IP in small dimension.} Problem~\eqref{IP} with $f$ a convex function; $S$ can be represented by polynomial inequalities, a first-order oracle, a separation oracle, or as a semialgebraic set.
\begin{ilps}
\task*[\textbf{$\top$}] $n^{\frac{4}{3}n} \langle B \r$, where $S$ is contained in a ball of radius $B$~\cite{DadushV:2013}.
\task*[\textbf{$+$}] Strictly stronger than ILP. Representing constraints implicitely by an oracle allows better dependence on instance size (cf. Model~\ref{mod:cds:convex}).
\task*[\textbf{$-$}] Exponential space. Algorithms usually impractical. Proving convexity can be difficult.
\task*[\textbf{$\heartsuit$}] \lv{Grötschel, Lovász, and Schrijver~\cite[Theorem 6.7.10]{GLS} (weak separation oracle), Khachiyan and Porkolab~\cite{KhachiyanP00} (semialgebraic sets), Heinz~\cite{Heinz05}, whose algorithm is superseded by Hildebrand and Köppe~\cite{HildebrandK13} (polynomials), Dadush, Peikert and Vempala~\cite{DadushV:2013} randomized and Dadush and Vempala~\cite{DadushV:2013} (strong separation oracle), Oertel, Wagner, and Weismantel~\cite{OertelWW14} reduction to Mixed ILP subproblems (first-order oracle).}\sv{\cite[Theorem 6.7.10]{GLS} (weak separation oracle), \cite{KhachiyanP00} (semialgebraic sets), \cite{Heinz05,HildebrandK13} (polynomials), \cite{DadushPV11} randomized /~\cite{DadushV:2013} deterministic (strong separation oracle), \cite{OertelWW14} reduction to Mixed ILP subproblems (first-order oracle).}
\task*[$\triangleright$] Hermelin et al. (multiagent scheduling; convex constraints)~\cite{HermelinKSTW:2015}, Bredereck et al. (bribery; convex objective)~\cite{BredereckFNST:2015}, Mnich and Wiese, Knop and Koutecký (scheduling; convex objective)~\cite{MnichW:2015,KnopK:2016}, Knop et al. (various problems; convex objectives)~\cite{KnopMT18}, Model~\ref{mod:cds:convex}
\end{ilps}

\tool{Indefinite quadratic IP in small dimension.} Problem~\eqref{LinIP} with $f(\vex) = \vex^{\intercal} Q \vex$ indefinite (non-convex) quadratic.
\begin{ilps}
\task*[\textbf{$\top$}] $g(n, \|A\|_\infty, \|Q\|_\infty) \langle \veb \r$~\cite{Zemmer:2017}
\task*[\textbf{$+$}] Currently the only tractable indefinite objective.
\task*[\textbf{$-$}] Limiting parameterization.
\task*[\textbf{$\heartsuit$}] \lv{Lokshtanov~\cite{Lokshtanov:2015}, Zemmer~\cite{Zemmer:2017}}\sv{\cite{Lokshtanov:2015,Zemmer:2017}}
\task*[\textbf{$\triangleright$}] Lokshtanov (\textsc{Optimal Linear Arrangement}~\cite{Lokshtanov:2015}, Model~\ref{mod:mc}
\end{ilps}

\tool{Parametric ILP in small dimension.}
Given a $Q = \{\veb \in \RR^m \mid B \veb \leq \ved\}$, decide
\[
\forall \veb \in Q \cap \ZZ^m \, \exists \vex \in \ZZ^n:\, A\vex \leq \veb \enspace .
\]
\begin{ilps}
\task*[\textbf{$\top$}] $g(n,m) \poly(\|A, B, \ved\|_\infty)$~\cite{EisenbrandShmonin2008}
\task*[\textbf{$+$}] Models one quantifier alternation. Useful in expressing game-like constraints (e.g., ``$\forall$ moves $\exists$ a counter-move''). Allows unary big-$M$ coefficients to model logic~\cite[Theorem 4.5]{KnopKM:2018}.
\task*[\textbf{$-$}] Input has to be given in unary (vs. e.g. Lenstra's algorithm).
\task*[\textbf{$\heartsuit$}] \lv{Eisenbrand and Shmonin~\cite[Theorem 4.2]{EisenbrandShmonin2008}, Crampton et al.~\cite[Corollary 1]{CramptonEtAl2017}}\sv{\cite[Theorem 4.2]{EisenbrandShmonin2008}, \cite[Corollary 1]{CramptonEtAl2017}}
\task*[\textbf{$\triangleright$}] Crampton et al. (resiliency)~\cite{CramptonEtAl2017}, Knop et al. (Dodgson bribery)~\cite{KnopKM:2018}
\end{ilps}

\subsection{Variable Dimension}
In this section it will be more natural to consider the following \emph{standard form} of \eqref{LinIP}
\begin{equation}
\min \{f(\vex)\mid A\vex = \veb, \, \vel \leq \vex \leq \veu, \, \vex \in \ZZ^n\}  \tag{SLinIP},\label{StandardLinIP}
\end{equation}
where $\veb \in \ZZ^m$ and $\vel, \veu \in \ZZ^n$.
Let $L = \langle f_{\max}, \veb, \vel, \veu \rangle$.
In contrast with the previous section, the following algorithms typically rely on algebraic arguments and dynamic programming.
\lv{
The large family of algorithms based on Graver bases (see below) can be described as \emph{iterative augmentation} methods, where we start with a feasible integer solution $\vex_0$ and iteratively find a step $\veg \in \{\vex \in \ZZ^n \mid A\vex = \mathbf{0}\}$ such that $\vex_0 + \veg$ is still feasible and improves the objective.
Under a few additional assumptions on $\veg$ it is possible to prove quick convergence of such methods.
}

\medskip
\tool{ILP with few rows.} Problem~\eqref{StandardLinIP} with small $m$ and a linear objective $\vew\vex$ for $\vew \in \ZZ^n$.
\begin{ilps}
\task*[\textbf{$\top$}] $\Oh((m\|A\|_\infty)^{2m}) \langle \veb \r$ if $\vel \equiv \mathbf{0}$ and $\veu\equiv +\mathbf{\infty}$, and $n \cdot (m \|A\|_\infty)^{\Oh(m^2)} \langle \veb, \vel, \veu \r$ in general~\cite{JansenR:2018}
\task*[\textbf{$+$}] Useful for configuration IPs with small coefficients, leading to exponential speed-ups. Best runtime in the case without upper bounds.
Linear dependence on $n$.
\task*[\textbf{$-$}] Limited modeling power. Requires small coefficients.
\task*[\textbf{$\heartsuit$}] \lv{Papadimitriou~\cite{Papadimitriou:1981}, Eisenbrand and Weismantel~\cite{EisenbrandW:2018}, Jansen and Rohwedder~\cite{JansenR:2018}}\sv{\cite{Papadimitriou:1981,EisenbrandW:2018,JansenR:2018}}
\task*[\textbf{$\triangleright$}] Jansen and Rohwedder (scheduling)~\cite{JansenR:2018}
\end{ilps}

\begin{minipage}{.45\textwidth}
\begin{align*}
A_{\text{nfold}} = \left(
\begin{array}{cccc}
  A_1    & A_1    & \cdots & A_1    \\
  A_2    & 0      & \cdots & 0      \\
  0      & A_2    & \cdots & 0      \\
  \vdots & \vdots & \ddots & \vdots \\
  0      & 0      & \cdots & A_2    \\
\end{array}
\right)
\end{align*}
\end{minipage}
\begin{minipage}{.45\textwidth}
\begin{align*}
A_{\text{stoch}} = \left(
\begin{array}{ccccc}
B_1    & B_2 & 0   & \cdots & 0    \\
B_1    & 0   & B_2   & \cdots & 0      \\
\vdots & \vdots & \vdots & \ddots & \vdots \\
B_1      & 0   & 0 & \cdots & B_2    \\
\end{array}
\right)
\end{align*}
\end{minipage}

\tool{$\mathbf{n}$-fold IP, tree-fold IP, and dual treedepth.}
\emph{$n$-fold IP} is problem~\eqref{StandardLinIP} in dimension $nt$, with $A = A_{\text{nfold}}$ for some two blocks $A_1 \in \ZZ^{r \times t}$ and $A_2 \in \ZZ^{s \times t}$, $\vel, \veu \in \ZZ^{nt}$, $\veb \in \ZZ^{r+ns}$, and with $f$ a separable convex function, i.e., $f(\vex) = \sum_{i=1}^{n} \sum_{j=1}^t  f_j^i(x_j^i)$ with each $f_j^i: \ZZ \to \ZZ$ convex.
\emph{Tree-fold IP} is a generalization of $n$-fold IP where the block $A_2$ is itself replaced by an $n$-fold matrix, and so on, recursively, $\tau$ times.
Tree-fold IP has bounded $\td_D(A)$.

\begin{ilps}
\task*[\textbf{$\top$}] $(\|A\|_\infty rs)^{\Oh(r^2s + rs^2)} (nt)^2 \log (nt) \langle L \r$ $n$-fold IP~\cite{AltmanovaKK:2018,EisenbrandHK:2018}; $(\|A\|_\infty+1)^{2^{\td_D(A)}} (nt)^2 \log (nt) \langle L \r$ for~\eqref{StandardLinIP}~\cite{KouteckyLO:2018}.
\task*[\textbf{$+$}] Variable dimension useful in modeling many ``types'' of objects~\cite{KnopKM:2017,KnopKM:2018}.
Useful for obtaining exponential speed-ups (not only configuration IPs).
Seemingly rigid format is in fact not problematic (blocks can be different provided coefficients and dimensions are small).
\task*[\textbf{$-$}] Requires small coefficients.
\task*[\textbf{$\heartsuit$}] \lv{Hemmecke et al.~\cite{HemmeckeOR13}, Knop et al.~\cite{KnopKM:2017b}, Chen and Marx~\cite{ChenM:2018}, Eisenbrand et al.~\cite{EisenbrandHK:2018}, Altmanová et al.~\cite{AltmanovaKK:2018}, Koutecký et al.~\cite{KouteckyLO:2018}}\sv{\cite{HemmeckeOR13,KnopKM:2017b,ChenM:2018,EisenbrandHK:2018,AltmanovaKK:2018,KouteckyLO:2018}}
\task*[\textbf{$\triangleright$}] Knop and Koutecký (scheduling with many machine types)~\cite{KnopK:2016}, Knop et al. (bribery with many voter types)~\cite{KnopKM:2017,KnopKM:2017b}, Chen and Marx (scheduling; tree-fold IP)~\cite{ChenM:2018}, Jansen et al. (scheduling EPTAS)~\cite{JansenKMR:2018}, Model~\ref{mod:sc:nfold}
\end{ilps}

\tool{2-stage and multi-stage stochastic IP, and primal treedepth.}
2-stage stochastic IP is problem~\eqref{StandardLinIP} with $A = A_{\text{stoch}}$ and $f$ a separable convex function; multi-stage stochastic IP is problem~\eqref{StandardLinIP} with a multi-stage stochastic matrix, which is the transpose of a tree-fold matrix; multi-stage stochastic IP is in turn generalized by IP with small primal treedepth $\td_P(A)$.
\begin{ilps}
\task*[\textbf{$\top$}] $g(\td_P(A), \|A\|_\infty) n^2 \log n \langle L \r$, $g$ computable~\cite{KouteckyLO:2018}
\task*[\textbf{$+$}] Similar to Parametric ILP in fixed dimension, but quantification $\forall \veb \in Q \cap \ZZ^n$ is now over a polynomial sized but possibly non-convex set of explicitely given right hand sides.
\task*[\textbf{$-$}] Not clear which problems are captured. Requires small coefficients. Parameter dependence $g$ is possibly non-elementary; no upper bounds on $g$ are known, only computability.
\task*[\textbf{$\heartsuit$}] \lv{Hemmecke and Schultz~\cite{HemmeckeS:2001}, Aschenbrenner and Hemmecke~\cite{AschenbrennerH:2007}, Koutecký et al.~\cite{KouteckyLO:2018}}\sv{\cite{HemmeckeS:2001,AschenbrennerH:2007,KouteckyLO:2018}}
\task*[\textbf{$\triangleright$}] N/A
\end{ilps}

\tool{Small treewidth and Graver norms.}
Let $g_\infty(A) = \max_{\veg \in \G(A)} \|\veg\|_\infty$ and $g_1(A) = \max_{\veg \in \G(A)} \|\veg\|_1$ be maximum norms of elements of $\G(A)$.
\begin{ilps}
\task*[\textbf{$\top$}] $\min \{ g_\infty(A)^{\Oh(\tw_P(A))}, g_1(A)^{\Oh(\tw_D(A))} \} n^2 \log n \langle L \r$~\cite{KouteckyLO:2018}
\task*[\textbf{$+$}] Captures IPs beyond the classes defined above (cf. Section~\ref{sec:sumcol:graver}).
\task*[\textbf{$-$}] Bounding $g_1(A)$ and $g_\infty(A)$ is often hard or impossible.
\task[\textbf{$\heartsuit$}] \lv{Koutecký et al.~\cite{KouteckyLO:2018}}\sv{\cite{KouteckyLO:2018}}
\task[\textbf{$\triangleright$}] Model~\ref{mod:sc:graver}
\end{ilps}

\section{Convex Constraints: Capacitated Dominating Set}\label{sec:convexConstraints}
\sv{\toappendix{\section{Additional Material to Section~\ref{sec:convexConstraints}}}}
\prob{\textsc{Capacitated Dominating Set}}
{A graph $G=(V,E)$ and a capacity function $c\colon V \to \N$.}
{%
Find a smallest possible set $D \subseteq V$ and a mapping $\delta\colon V \setminus D \to D$ such that for each $v \in D$, $|\delta^{-1}(v)| \leq c(v)$.
}

\subparagraph*{Solution Structure.}
Let $<_c$ be a linear extension of ordering of $V$ by vertex capacities, i.e., $u <_c v$ if $c(u) \leq c(v)$.
For $i \in T(G)$ and $\ell \in [|V_i|]$ let $V_i[1:\ell]$ be the set of the first $\ell$ vertices of $V_i$ in the ordering $<_c$ and let $f_i(\ell) = \sum_{v \in V_i[2:\ell]} c(v)$; for $\ell > |V_i|$ let $f_i(\ell) = f_i(|V_i|)$.
Let $D$ be a solution and $D_i = D \cap V_i$.
We call the functions $f_i$ the \emph{domination capacity functions}.
Intuitively, $f_i(\ell)$ is the maximum number of vertices dominated by $V_i[1:\ell]$.
Observe that since $f_i(\ell)$ is a partial sum of a non-increasing sequence of numbers, it is a piece-wise linear concave function.
We say that $D$ is \emph{capacity-ordered} if, for each $i \in T(G)$, $D_i = V_i[1:|D_i|]$.
The following observation allows us to restrict our attention to such solutions; the proof goes by a simple exchange argument.
\sv{\begin{lemma}[\appmark]\label{lem:CDScapacityOrientedOPT}}
\lv{\begin{lemma}\label{lem:CDScapacityOrientedOPT}}
There is a capacity-ordered optimal solution.
\end{lemma}
\toappendix{
\begin{proof}\sv{[Proof of Lemma~\ref{lem:CDScapacityOrientedOPT}]}
Consider any solution $D$ together with a mapping $\delta: V \setminus D \to D$ witnessing that $D$ is a solution.
Our goal is to construct a capacity-ordered solution $\hat{D}$ which is at least as good as $D$.
If $D$ itself is capacity-ordered, we are done.
Assume the contrary; thus, there exists an index $i \in [k]$ and a vertex $v \in D_i$ such that $v \not\in V_i[1:|D_i|]$, and consequently there exists a vertex $u \in V_i[1:|D_i|]$ such that $u \not\in D_i$.

Let $D' \subseteq V$ be defined by setting $D'_i = (D_i \cup \{u\}) \setminus \{v\}$ and $D'_j = D_j$ for each $j \neq i$.
We shall define a mapping $\delta'$ witnessing that $D'$ is again a solution.
Let $\delta'(x) = y$ iff $\delta(x) = y$ and $x \neq u$ and $y \neq v$, let $\delta'(x) = u$ whenever $\delta(x) = v$ and let $\delta'(v) = y$ if $\delta(u) = y$.
Clearly $|(\delta')^{-1}(x)| \leq c(x)$ for each $x \in D$ because $|(\delta')^{-1}(x)| = |\delta^{-1}(x)|$ when $x \not\in \{u,v\}$, and $|(\delta')^{-1}(u)| = |\delta^{-1}(v)|$ and $c(u) \geq c(v)$.

If $D'$ itself is not yet a capacity-ordered solution, we repeat the same swapping argument.
Observe that \begin{enumerate*}
\item $\sum_{i=1}^k |D_i \triangle V_i[1:|D_i|] > \sum_{i=1}^k |D'_i \triangle V_i[1:|D'_i|]$, i.e., $D'$ is closer than $D$ to being capacity-ordered, and,
\item the size of $D'$ compared to $D$ does not increase.
\end{enumerate*}
Finally, when $\sum_{i=1}^k |D'_i \triangle V_i[1:|D'_i|] = 0$, $D'$ is our desired capacity-ordered solution $\hat{D}$.
\end{proof}
}

\vskip -.2cm
Observe that a capacity-ordered solution is fully determined by the sizes $|D_i|$ and $<_c$ rather than the actual sets $D_i$, which allows modeling CDS in small dimension.

\begin{model}[\textsc{Capacitated Dominating Set} as convex IP in fixed dimension]\label{mod:cds:convex}~\\
\textbf{Variables \& notation:}
\begin{tasks}[style=itemize](2)
		\task $x_i = |D_i|$
		\task $y_{ij} = |\delta^{-1}(D_i) \cap D_j|$
		\task* $f_i(x_i)$ = maximum \#vertices dominated by $D_i$ if $|D_i|=x_i$
	\end{tasks}
\textbf{Objective \& Constraints:}
\begin{align}
\min & \sum_{i \in T(G)} x_i & & & \comment{\min |D| = \sum_{i \in T(G)} |D_i|} \tag{cds:cds-obj} \label{eqn:cds_start}  \\
\sum_{j \in N_{T(G)}(i)} y_{ij} &\leq f_i(x_i) & \forall i \in T(G) & & \comment{\text{respect capacities}} \tag{cds:cap} \label{eqn:cds_convex}\\
\sum_{i \in N_{T(G)}(j)} y_{ij} &\geq |V_j| - x_j & \forall j \in T(G) & & \comment{\text{every $v \in V_j \setminus D_j$ dominated}} \tag{cds:dom} \\
0 \leq x_i &\leq |V_i| & \forall i \in T(G) & & \tag{cds:bounds} \label{eqn:cds_end}
\end{align}
\textbf{Parameters \& Notes:}
\vspace{-.6em}
\begin{itemize}
\item
\begin{tabular}{c c c c c c }
\#vars & \#constraints & $\|A\|_\infty$ & $\|\veb\|_\infty$ & $\|\vel,\veu\|_\infty$ & $\|\vew\|_\infty$ \\
$\Oh(k^2)$ & $\Oh(k)$ & $1$ & $|G|$ & $|G|$ & $1$
\end{tabular}
\item constraint~\eqref{eqn:cds_convex} is convex, since it bounds the area \emph{under} a concave function, and is piece-wise linear. \qed
\end{itemize}
\end{model}

\toappendix{
\lv{\begin{figure}[bt]}
\sv{\begin{figure}[!h]}
\begin{minipage}[b]{.47\textwidth}
\includegraphics[width=\textwidth]{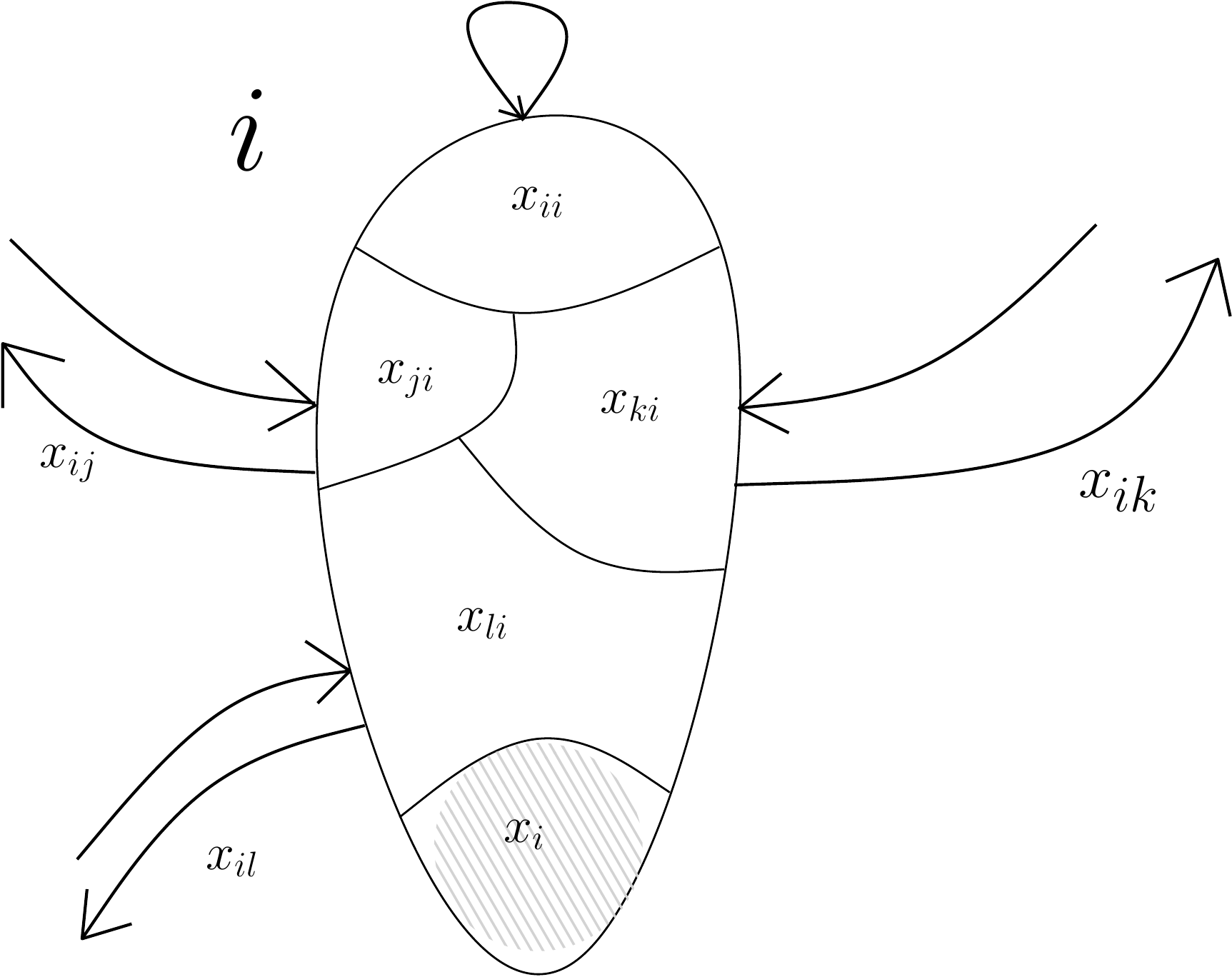} \label{fig:cds}
\end{minipage}
\begin{minipage}[b]{.47\textwidth}
\includegraphics[width=\textwidth]{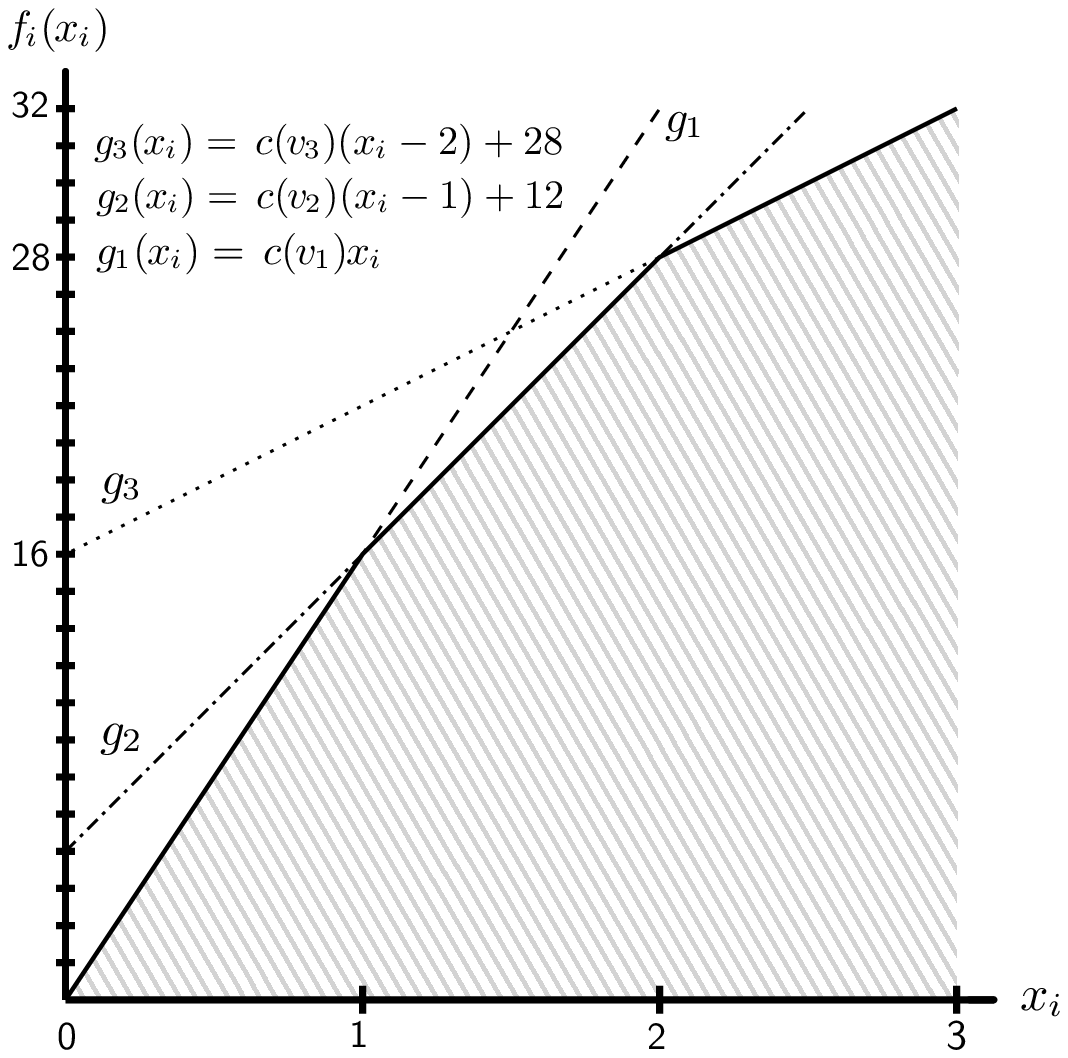}
\end{minipage}

\begin{minipage}[t]{.47\textwidth}
\caption{\label{fig:interpretationForCDS}%
Interpretation of variables of Model~\ref{mod:cds:convex}.}
\end{minipage}
\begin{minipage}[t]{.47\textwidth}
\caption{\label{fig:linearization}%
The linearization~\eqref{eqn:cds:cap-lin} of a piecewise linear convex constraint~\eqref{eqn:cds_convex} in Model~\ref{mod:cds:linear}.}
\end{minipage}
\end{figure}
}

\vskip -.5cm
Then, applying for example Dadush's algorithm~\cite{DadushPV11} to Model~\ref{mod:cds:convex} yields Theorem~\ref{thm:cds}\ref{thm:cds:convex}.
We can trade the non-linearity of the previous model for an increase in the number of constraints and the largest coefficient.
That, combined with Lenstra's algorithm, yields Theorem~\ref{thm:cds}\ref{thm:cds:ilp}, where we get a larger dependence on $|G|$, but require only $\poly(k,|G|)$ space.

\begin{model}[\textsc{Capacitated Dominating Set} as ILP in fixed dimension] \label{mod:cds:linear}~\\
Exactly as Model~\ref{mod:cds:convex} but replace constraints~\eqref{eqn:cds_convex} with the following equivalent set of $|G|$ linear constraints:
\begin{align}
\sum_{ij \in E(T(G))} y_{ij} &\leq f_i(\ell-1) + c(v_\ell)(x_i - \ell + 1) & \forall i \in T(G)\, \forall \ell \in [|V_i|] \tag{cds:cap-lin} \label{eqn:cds:cap-lin}
\end{align}
The parameters then become:
\begin{tabular}{c c c c c c }
\#vars & \#constraints & $\|A\|_\infty$ & $\|\veb\|_\infty$ & $\|\vel,\veu\|_\infty$ & $\|\vew\|_\infty$ \\
$\Oh(k^2)$ & $\Oh(k+|G|)$ & $|G|$ & $|G|$ & $|G|$ & $1$
\end{tabular} \qed \vspace{1em}
\end{model}

%

\begin{proof}[{[Additive approximation]} Proof of Theorem~\ref{thm:cds}\ref{thm:cds:approx}]
Let $(\vex, \vey) \in \RR^{k + k^2}$ be an optimal solution to the \emph{continuous relaxation} of Model~\ref{mod:cds:convex}, i.e., we relax the requirement that $(\vex, \vey)$ are integral; note that such $(\vex, \vey)$ can be computed in polynomial time using the ellipsoid method~\cite{GLS}, or by applying a polynomial LP algorithm to Model~\ref{mod:cds:linear}.
We would like to round $(\vex, \vey)$ up to an integral $(\hat{\vex}, \hat{\vey})$ to obtain a feasible integer solution which would be an approximation of an integer optimum.
Ideally, we would take $\hat{\vey} = \lceil \vey \rceil$ and compute $\hat{\vex}$ accordingly, i.e., set $\hat{x}_i$ to be smallest possible such that $\sum_{j \in N_{T(G)}(i)} \hat{y}_{ij} \geq f_i(\hat{x}_i)$; note that $\hat{x}_i \leq x_i + k$, since we add at most $k$ neighbors (to be dominated) in neighborhood of $V_i$.
However, this might result in a non-feasible solution if, for some $i$, $\hat{x}_i > |V_i|$.
In such a case, we solve the relaxation again with an additional constraint $x_i = |V_i|$ and try rounding again, repeating this aforementioned fixing procedure if rounding fails, and so on.
After at most $k$ repetitions this rounding results in a feasible integer solution $(\hat{\vex}, \hat{\vey})$, in which case we have $\|\hat{\vex} - \vex\|_1 \leq k^2$ and thus the solution represented by $(\hat{\vex}, \hat{\vey})$ has value at most $OPT+k^2$; the relaxation must eventually become feasible as setting $x_i=|V_i|$ for all $i\in T(G)$ yields a feasible solution.
\end{proof}

\begin{proof}[{[Speed trade-offs]} Proof of Theorem~\ref{thm:cds}\ref{thm:cds:speed}]
Notice that on our way to proving Theorem~\ref{thm:cds}\ref{thm:cds:approx} we have shown that Model~\ref{mod:cds:convex} has \emph{integrality gap} at most $k^2$, i.e., the value of the continuous optimum is at most $k^2$ less than the value of the integer optimum.
This implies that an integer optimum $(\vex^*, \vey^*)$ satisfies, for each $i \in [k]$, $\max \{0, \lfloor x_i - k^2 \rfloor\} \leq x^*_i \leq \min \{|V_i|, x_i + \lceil k^2 \rceil \}$.

We can exploit this to improve Theorem~\ref{thm:cds}\ref{thm:cds:convex} in terms of the parameter dependence at the cost of the dependence on $|G|$.
Let us assume that we have a way to test, for a given integer vector $\hat{\vex}$, whether it models a capacity-ordered solution, that is, whether there exists a capacitated dominating set with $D_i = V_i[1:\hat{x}_i]$ for each $i$.
Then we can simply go over all possible $(2k^2 + 2)^{k}$ choices of $\hat{\vex}$ and choose the best.
So we are left with the task of, given a vector $\hat{\vex}$, deciding if it models a capacity-ordered solution.

But this is easy.
Let $<_c$ be the assumed order and define $D$ as above.
Now, we construct an auxiliary bipartite matching problem, where we put $c(v)$ copies of each vertex from $D$ on one side of the graph, and all vertices of $V \setminus D$ on the other side, and connect a copy of $v \in D$ to $u \in V \setminus D$ if $uv \in E(G)$.
Then, $D$ is a capacitated dominating set if and only if all vertices in $V \setminus D$ can be matched.
The algorithm is then simply to compute the continuous optimum $\vex$, and go over all integer vectors $\hat{\vex}$ with $\|\vex - \hat{\vex}\|_1 \leq k^2$, verifying whether they model a solution and choosing the smallest (best) one.
\end{proof}


\section{Indefinite Quadratics: Max $q$-Cut}
\prob{\textsc{Max-$q$-Cut}}
{A graph $G=(V,E)$.}
{A partition $W_1 \dot{\cup} \cdots \dot{\cup} W_q = V$ maximizing the number of edges between distinct $W_\alpha$ and $W_\beta$, i.e., $|\{uv \in E(G) \mid u\in W_\alpha, v \in W_\beta, \alpha \neq \beta\}|$. \vspace{0.25em}  }
\subparagraph*{Solution structure.}
As before, it is enough to describe \emph{how many} vertices from type $i \in T(G)$ belong to $W_\alpha$ for $\alpha \in [q]$, and their specific choice does not matter; this gives us a small dimensional encoding of the solutions.

\begin{model}[\textsc{Max-$q$-Cut} as \textsc{LinIP} with indefinite quadratic objective]\label{mod:mc}
~\\ \textbf{Variables \& Notation:}
\begin{tasks}[style=itemize](2)
\task $x_{i\alpha} = |V_i \cap W_\alpha|$
\task $x_{i\alpha} \cdot x_{j \beta} = $ \#edges between $V_i \cap W_\alpha$ and $V_j \cap W_\beta$ if $ij \in E(T(G))$.
\end{tasks}
\textbf{Objective \& Constraints:}
\begin{align}
\min & \sum_{\substack{\alpha, \beta \in [q]:\\ \alpha \neq \beta}} \sum_{ij \in E(T(G))} x_{i\alpha} \cdot x_{j\beta} & & & \comment{\min \text{\#edges across partites}} \tag{mc:obj} \label{eqn:mc:obj}\\
\sum_{\alpha \in [q]} x_{i\alpha} &= |V_i| & \forall i \in T(G) & & \comment{\text{$\left(V_i \cap W_\alpha\right)_{\alpha \in [q]}$ partitions $V_i$}} \tag{mc:part} \label{eqn:mc:part}
\end{align}
\textbf{Parameters \& Notes:}
\begin{itemize}
\item
\begin{tabular}{c c c c c c }
\#vars & \#constraints & $\|A\|_\infty$ & $\|\veb\|_\infty$ & $\|\vel,\veu\|_\infty$ & $\|Q\|_\infty$ \\
$kq$ & $k$ & $1$ & $|G|$ & $|G|$ & $1$
\end{tabular}
\item objective~\eqref{eqn:mc:obj} is indefinite quadratic. \qed
\end{itemize}
\end{model}
Applying Lokshtanov's~\cite{Lokshtanov:2015} or Zemmer's~\cite{Zemmer:2017} algorithm to Model~\ref{mod:mc} yields Theorem~\ref{thm:maxcut}.
Note that since we do not know anything about the objective except that it is quadratic, we have to make sure that $\|Q\|_\infty$ and $\|A\|_\infty$ are small.

\section{Convex Objective: \textsc{Sum Coloring}}\label{sec:SumColoring}
\sv{\toappendix{\section{Additional Material to Section~\ref{sec:SumColoring}}}}
\prob{\textsc{Sum Coloring}}
{A graph $G=(V,E)$.}
{A proper coloring $c\colon V \to \NN$ minimizing $\sum_{v \in V} c(v)$.}

In the following we first give a single-exponential algorithm for \textsc{Sum Coloring} with a polynomial dependence on $|G|$, then a double-exponential algorithm with a logarithmic dependence on $|G|$, and finally show how to combine the two ideas together to obtain a single-exponential algorithm with a logarithmic dependence on $|G|$.

\subsection{\textsc{Sum Coloring} via $n$-fold IP}
\subparagraph*{Structure of Solution.}
The following observation was made by Lampis~\cite{Lampis12} for the \textsc{Coloring} problem, and it holds also for the \textsc{Sum Coloring} problem: every color $C \subseteq V(G)$ intersects each clique type in at most one vertex, and each independent type in either none or all of its vertices.
The first follows simply by the fact that it is a clique; the second by the fact that if both colors $\alpha,\beta$ with $\alpha < \beta$ are used for an independent type, then recoloring all vertices of color $\beta$ to be of color $\alpha$ remains a valid coloring and decreases its cost.
We call a coloring with this structure an \emph{essential coloring}.

\begin{model}[\textsc{Sum Coloring} as $n$-fold IP]\label{mod:sc:nfold}
~\\ \textbf{Variables \& Notation:}
\begin{tasks}[style=itemize](2)
\task $x_i^\alpha = 1$ if color $\alpha$ intersects $V_i$
\task $\alpha \cdot x_i^\alpha = $ cost of color $\alpha$ at a clique type $i$
\task* $\alpha |V_i| \cdot x_i^\alpha = $ cost of color $\alpha$ at an independent type $V_i$
\task* $S_{\text{nfold}}(\vex) = \sum_{\alpha=1}^{|G|} \left( (\sum_{\text{clique } i \in T(G)} \alpha x_i^\alpha) + (\sum_{\text{indep. } i \in T(G)} \alpha |V_i| x^\alpha_i) \right) = $ total cost of $\vex$
\end{tasks}
\textbf{Objective \& Constraints:}
\begin{align}
\min ~& S_{\text{nfold}}(\vex) & & & \tag{sc:nf:obj} \label{eqn:sumcol_nfold_obj}\\
\sum_{\alpha=1}^{|G|} x_i^\alpha &= |V_i| & \forall i \in T(G), \text{$V_i$ is clique} & & \comment{\text{$V_i$ is colored}} \tag{sc:nf:cliques} \label{eqn:sc:nfold:clique}\\
\sum_{\alpha=1}^{|G|} x_i^\alpha &= 1 & \forall i \in T(G), \text{$V_i$ is independent} & & \comment{\text{$V_i$ is colored}} \tag{sc:nf:indeps}\label{eqn:sc:nfold:indep} \\
x_i^\alpha + x_j^\alpha & \leq 1 & \forall \alpha \in \left[ |G| \right]\, \forall ij \in E(T(G)) & & \comment{\vex^\alpha \text{ is independent set}} \tag{sc:nf:xi-indep}\label{eqn:sumcol_nfold_col}
\end{align}
\textbf{Parameters \& Notes:}
\vspace{-.6em}
\begin{itemize}
\item
\begin{tabular}{c c c c c c c c c}
\#vars & \#constraints & $\|A\|_\infty$ & $\|\veb\|_\infty$ & $\|\vel,\veu\|_\infty$ & $\|\vew\|_\infty$ & $r$ & $s$ & $t$\\
$k|G|$ & $k+k^2|G|$ & $1$ & $|G|$ & $1$ & $|G|$ & $k$ & $k^2$ & $k$
\end{tabular}
\item Constraints have an $n$-fold format: \eqref{eqn:sc:nfold:clique} and~\eqref{eqn:sc:nfold:indep} form the $(A_1 \cdots A_1)$ block and~\eqref{eqn:sumcol_nfold_col} form the $A_2$ blocks; see parameters $r,s,t$ above.
\lv{Observe that the matrix $A_1$ is the $k \times k$ identity matrix and the matrix $A_2$ is the incidence matrix of $T(G)$ transposed.}
\qed
\end{itemize}
\end{model}
Applying the algorithm of Altmanová et al.~\cite{AltmanovaKK:2018} to Model~\ref{mod:sc:nfold} yields Theorem~\ref{thm:sumcol}\ref{thm:sumcol:nfold}.
Model~\ref{mod:sc:nfold} is a typical use case of $n$-fold IP: we have a vector of multiplicities $\veb$ (modeling $(|V_1|, \dots, |V_k|)$) and we optimize over its decompositions into independent sets of $T(G)$.
A clever objective function models the objective of \textsc{Sum Coloring}.
\lv{The main drawback is large number of bricks in this model.}

\subsection{\textsc{Sum Coloring} via Convex Minimization in Fixed Dimension}
\subparagraph*{Structure of Solution.}
The previous observations also allow us to encode a solution in a different way.
Let $\II = \{I_1, \dots, I_K\}$ be the set of all independent sets of $T(G)$; note that $K < 2^k$.
Then we can encode an essential coloring of $G$ by a vector of multiplicities $\vex=(x_{I_1}, \dots, x_{I_K})$ of elements of $\II$ such that there are $x_{I_j}$ colors which color exactly the types contained in $I_j$.
The difficulty with \textsc{Sum Coloring} lies in the formulation of its objective function.
Observe that given an $I \in \II$, the number of vertices every color class of this type will contain is independent of the actual multiplicity $x_I$.
Define the {\em size of a color class} $\sigma\colon \II\to \NN$ as
$\sigma(I)=\sum_{\text{clique } i \in I}1+\sum_{\text{indep. } i \in I}|V_i|$.

\sv{\begin{lemma}[\appmark]\label{lem:order_I}}
\lv{\begin{lemma}\label{lem:order_I}}
Let $G = (V,E)$ be a graph and let $c\colon V \to \NN$ be a proper coloring of $G$ minimizing $\sum_{v\in V} c(v)$.
Let $\mu(p)$ denote the quantity $\left|\left\{ v \in V \mid c(v) = p \right\}\right|$.
Then $\mu(p) \ge \mu(q)$ for every $p \le q$.
\end{lemma}
\toappendix{
\begin{proof}\sv{[Proof of Lemma~\ref{lem:order_I}]}
Suppose for contradiction that we have $p < q$ with $\mu(p) < \mu(q)$.
We now construct a proper coloring $c'$ of $G$ as follows
\[
c'(v) = \begin{cases}
p & \textrm{if } c(v) = q, \\
q & \textrm{if } c(v) = p, \\
c(v) & \textrm{otherwise}.
\end{cases}
\]
Clearly $c'$ is a proper coloring.
Now we have
\begin{align*}
\sum_{v\in V} c(v) = &
\left( \sum_{v\in V} c'(v) \right) - p\mu(q) - q\mu(p) + p\mu(p) + q\mu(q) = \\
&\left( \sum_{v\in V} c'(v) \right) - p (\mu(q) - \mu(p)) + q (\mu(q) - \mu(p)) = \\
&\left( \sum_{v\in V} c'(v) \right) + (\mu(q) - \mu(p))(q - p) > \sum_{v\in V} c'(v) \,.
\end{align*}
Here the last inequality holds, since both the factors following the sum are positive due to our assumptions.
Thus we arrive at a contradiction that $c$ is a coloring minimizing the first sum.
\end{proof}
}

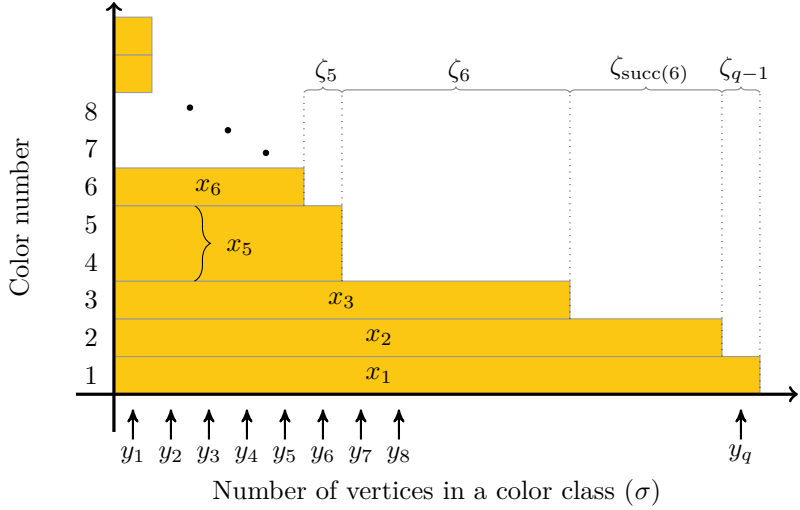
\begin{figure}[bt]
\begin{minipage}{.76\textwidth}
  \begin{tikzpicture}
  \tikzstyle{tecka}=[draw,circle,black,fill=black,inner sep=0pt,minimum size=2pt]
  \tikzstyle{var}=[black,pos=.5]
  \tikzstyle{cislo}=[black]
  \tikzstyle{axis}=[very thick,black,->]
  \tikzstyle{heavy var}=[fill=orange!60,orange!60,draw=black!40]
  \definecolor{lipicsYellow}{rgb}{0.99,0.78,0.07}
  \tikzstyle{light var}=[fill=lipicsYellow!100,lipicsYellow!100,draw=black!40]
  \tikzstyle{pomocnaCara}=[black!80,dotted]
  \tikzstyle{malaSvorka}=[decorate,decoration={brace,amplitude=2pt},black!50]

  \draw[light var] (0,0) rectangle (8.5,.5) node[var,xshift=-.76cm] {$x_1$};
  \draw[light var] (0,.5) rectangle (8,1) node[var,xshift=-.51cm] {$x_2$};
  \draw[light var] (0,1) rectangle (6,1.5) node[var] {$x_3$};

  \draw[light var] (0,1.5) rectangle (3,2.5) node[pos=.5] {};
  \draw [decorate,decoration={brace,amplitude=6pt,mirror},xshift=-4pt,yshift=0pt]
    (1.2,1.5) -- (1.2,2.5) node [black,midway,yshift=-.05cm,xshift=.6cm] 
    {$x_5$};
  \draw[light var] (0,2.5) rectangle (2.5,3) node[var] {$x_6$};

  \node[tecka] at (2,3.2) {};
  \node[tecka] at (1.5,3.5) {};
  \node[tecka] at (1,3.8) {};

  \draw[light var] (0,4) rectangle (0.5,4.5) node[pos=.5] {};
  \draw[light var] (0,4.5) rectangle (0.5,5) node[pos=.5] {};


  \draw[axis] (-.5,0) -- coordinate (x axis mid) (9,0);
  \draw[axis] (0,-.5) -- coordinate (y axis mid) (0,5.2);


  \node[rotate=90, above=.5cm, yshift=.5cm] at (y axis mid) {Color number};
  \node[below=1cm] at (x axis mid) {Number of vertices in a color class ($\sigma$)};

  \foreach \i in {1,...,8} {
    \pgfmathsetmacro{\y}{\i*0.5-0.25}
    \node[cislo] at (-.3,\y) {\i};
  }
  
  \foreach \i in {1,...,8} {
    \pgfmathsetmacro{\x}{\i*0.5-0.25}
    \node[cislo] at (\x,-.8) {$y_{\i}$};
    \draw[->,thick,>=stealth'] (\x,-.6) -- (\x,-.2);
  }
  \node[cislo] at (8.25,-.8) {$y_q$};
  \draw[->,thick,>=stealth'] (8.25,-.6) -- (8.25,-.2);

  \draw[pomocnaCara] (8.5,0) -- (8.5,4);
  \draw[pomocnaCara] (8,.5) -- (8,4);
  \draw[pomocnaCara] (6,1) -- (6,4);
  \draw[pomocnaCara] (3,1.5) -- (3,4);
  \draw[pomocnaCara] (2.5,2.5) -- (2.5,4);

  \draw[malaSvorka] (8,4) -- (8.5,4) node [black,midway,yshift=.3cm,xshift=.05cm] {$\zeta_{q-1}$};
  \draw[malaSvorka] (6,4) -- (8,4) node [black,midway,yshift=.3cm,xshift=.05cm] {$\zeta_{\Succ(6)}$};
  \draw[malaSvorka] (3,4) -- (6,4) node [black,midway,yshift=.3cm,xshift=.05cm] {$\zeta_6$};
  \draw[malaSvorka] (2.5,4) -- (3,4) node [black,midway,yshift=.3cm,xshift=.05cm] {$\zeta_{5}$};
\end{tikzpicture}
  \end{minipage}
  \begin{minipage}{.22\textwidth}
  \caption{\label{fig:fig_1}%
  An illustration of the cost decomposition to the individual classes.
  Note that $i$-th row (color $i$) has cost $i$ per vertex.}
  \end{minipage}
\end{figure}

Our goal now is to show that the objective function can be expressed as a convex function in terms of the variables $\vex$.
We will get help from auxiliary variables $y_1, \dots, y_{|G|}$ which are a linear projection of variables $\vex$; note that we do not actually introduce these variables into the model and only use them for the sake of proving convexity.
Namely, $y_j$ indicates how many color classes contain at least $j$ vertices:
$y_j=\sum_{\sigma(I)\geq j} x_I$.
Then, the objective function can be expressed as $S_{\text{convex}}(\vex)=\sum_{i=1}^{p} \left| i\sigma(I_i) \right| = \sum_{j=1}^{|G|}\binom{y_j}{2}$,
where $i=1,\dots,p$ is the order of the color classes given by Lemma~\ref{lem:order_I}, every class of type $I$ is present $x_I$ times, where we enumerate only those $I$ with $x_I \geq 1$.
The equivalence of the two is straightforward to check.

Finally, $S_{\text{convex}}$ is convex with respect to $\vex$ because,
\sv{
\begin{enumerate*}[label={\arabic*)},font={\bfseries}]
  \item all $x_I$ are linear (thus affine) functions,
  \item $y_i = \sum_{I: \sigma(I) \geq i} x_I$ is a sum of affine functions, thus affine,
  \item $y_i(y_i - 1) / 2$ is convex: it is a basic fact that $h(x) = g(f(x))$ is convex if $f$ is affine and $g$ is convex. Here $f = y_i$ is affine by the previous point and $g = f(f-1)/2$ is convex.
  \item $S_{\text{convex}}$ is the sum of $y_i(y_i - 1) / 2$, which are convex by the previous point.
\end{enumerate*}}
\lv{\begin{itemize}
  \item all $x_I$ are linear (thus affine) functions,
  \item $y_i = \sum_{I: \sigma(I) \geq i} x_I$ is a sum of affine functions, thus affine,
  \item $y_i(y_i - 1) / 2$ is convex: it is a basic fact that $h(x) = g(f(x))$ is convex if $f$ is affine and $g$ is convex. Here $f = y_i$ is affine by the previous point and $g = f(f-1)/2$ is convex.
  \item $S_{\text{convex}}$ is the sum of $y_i(y_i - 1) / 2$, which are convex by the previous point.
\end{itemize}}

\begin{model}[\textsc{Sum Coloring} as \textsc{LinIP} in fixed dimension with convex objective]\label{mod:sc:convex}
~\\ \textbf{Variables \& Notation:}
\begin{tasks}[style=itemize](2)
\task $x_I=$ \#of color class $I$
\task $y_i=$ \#of color classes $I$ with $\sigma(I) \leq i$
\task $\binom{y_i}{2}$ cost of column $y_i$ (Figure~\ref{fig:fig_1})
\task $S_{\text{convex}} = \sum_{i=1}^{|G|} \binom{y_i}{2} = $ cost of all columns
\end{tasks}
\textbf{Objective \& Constraints:}
\vspace{-.5em}
\begin{align}
\min ~&S_{\text{convex}}(\vex)  & & & \tag{sc:convex:obj} \label{sc:convex:obj}\\
\sum_{I_j: i \in I_j} x_{I_j} &= |V_i| & \forall \text{clique } i \in T(G) & & \comment{\text{clique $V_i$ gets $|V_i|$ colors}} \tag{sc:convex:cliques} \label{sc:convex:cliques} \\
\sum_{I_j: i \in I_j} x_{I_j} &= 1 & \forall \text{indep. } i \in T(G) & & \comment{\text{indep. $V_i$ gets $1$ color}} \tag{sc:convex:indeps} \label{sc:convex:indeps}
\end{align}
\textbf{Parameters \& Notes:}
\vspace{-.5em}
\begin{itemize}
\item
\begin{tabular}{c c c c c c}
\#vars & \#constraints & $\|A\|_\infty$ & $\|\veb\|_\infty$ & $\|\vel,\veu\|_\infty$ & $f_{\max}$ \\
$2^k$ & $k$ & $1$ & $|G|$ & $|G|$ & $|G|^2$
\end{tabular}
\item Objective $S_{\text{convex}}$ is non-separable convex, and can be computed in time $2^k \log |G|$ by noticing that there are at most $2^k$ different $y_i$'s (see below). \qed
\end{itemize}
\end{model}
Applying the algorithm of Dadush~\cite{DadushPV11} to Model~\ref{mod:sc:convex} yields Theorem~\ref{thm:sumcol}\ref{thm:sumcol:convex}.
Notice that we could not apply Lokshtanov's algorithm because the objective has large coefficients.
Also notice that we do not need separability of $S_{\text{convex}}$ or any structure of $A$.

\subsection{\textsc{Sum Coloring} and Graver Bases} \label{sec:sumcol:graver}
Consider Model~\ref{mod:sc:convex}.
The fact that the number of rows and the largest coefficient $\|A\|_\infty$ is small, and that we can formulate $S_{\text{convex}}$ as a separable convex objective in terms of the $y_i$ variables gives us some hope that Graver basis techniques would be applicable.

Since $|\II| \leq 2^k$, we can replace the $y_i$'s by a smaller set of variables $z_i$ for a set of ``critical sizes'' $\Gamma = \{i \in [|G|] \mid \exists I \in \II: \sigma(I) = i\}$.
For each $i \in \Gamma$ let $\Succ(i) = \min \{j \in \Gamma \mid j > i\}$ (and let $\Succ(\max \Gamma)=\max \Gamma$), define $z_i = \sum_{I \in \II: \sigma(I) \geq i} x_I$, and let $\zeta_i = (\Succ(i) - i)$ be the size difference between a color class of size $i$ and the smallest larger color class. Then,
\[
S_{\text{convex}}(\vex) = \sum_{i=1}^{|G|} \binom{y_i}{2} = \sum_{i \in \Gamma} \zeta_i \binom{z_i}{2} = S_{\text{sepconvex}}(\vez) \enspace .
\]

Now we want to construct a system of inequalities of bounded dual treewidth $\tw_D(A)$; however, adding the $z_i$ variables as we have defined them amounts to adding many inequalities containing the $z_1$ variable, thus increasing the dual treewidth to $k+2^k$.
To avoid this, let us define $z_i$ equivalently as $z_i = z_{\Succ(i)} + \sum_{\substack{I \in \II:\\ \Succ(i) > \sigma(I) \geq i}} x_I = z_{\Succ(i)} + \sum_{\substack{I \in \II:\\ \sigma(I) = i}} x_I$.
\lv{The last equality follows from the definition of $\Gamma$ which implies there are no independent sets with size strictly between $i$ and $\Succ(i)$ in $G$.}

\begin{model}[\textsc{Sum Coloring} as \textsc{LinIP} with small $\tw_D(A)$ and small $g_1(A)$]\label{mod:sc:graver}
~\\ \textbf{Variables \& Notation:}
\begin{tasks}[style=itemize](2)
\task $x_I=$ \#of color class $I$
\task $z_i=$ \#of color classes $I$ with $\sigma(I) \geq i$
\task* $\zeta_i =$ size difference between $I \in \II$ with $\sigma(I)=i$ and closest larger $J \in \II$
\task* $\zeta_i \binom{z_i}{2}$ cost of all columns between $y_i$ and $y_{\Succ(i)}$ (Figure~\ref{fig:fig_1})
\task $\Gamma = $ set of critical sizes
\task $S_{\text{sepconvex}}(\vez) = \sum_{i \in \Gamma} \zeta_i \binom{z_i}{2} = $ total cost
\end{tasks}
\textbf{Objective \& Constraints:} constraints~\eqref{sc:convex:cliques} and~\eqref{sc:convex:indeps}, and:
\vspace{-.5em}
\begin{align}
\min ~&S_{\text{sepconvex}}(\vez)  & & & \tag{sc:graver:obj} \label{sc:graver:obj}\\
z_i &= z_{\Succ(i)} + \sum_{I \in \II:\sigma(I) = i} x_I & \forall i \in \Gamma & & \tag{sc:graver:sep} \label{sc:graver:sep}
\end{align}
\textbf{Parameters \& Notes:}
\vspace{-.5em}
\begin{itemize}
\item
\begin{tabular}{c c c c c c c c}
\#vars & \#constraints & $\|A\|_\infty$ & $\|\veb\|_\infty$ & $\|\vel,\veu\|_\infty$ & $f_{\max}$ & $g_1(A)$ & $\tw_D(A)$ \\
$\OhOp{2^k}$ & $\OhOp{2^k}$ & $1$ & $|G|$ & $|G|$ & $|G|^2$ & $\OhOp{k^{k}}$  & $k+2$
\end{tabular}
\item Bounds on $g_1(A)$ and $\tw_D(A)$ by Lemmas~\ref{lem:graver} and~\ref{lem:tw}, respectively.
\item Objective $S_{\text{sepconvex}}$ is separable convex. \qed
\end{itemize}
\end{model}
Applying the algorithm of Koutecký et al.~\cite{KouteckyLO:2018} to Model~\ref{mod:sc:graver} yields Theorem~\ref{thm:sumcol}\ref{thm:sumcol:graver}.

Let us denote the matrix encoding the constraints~\eqref{sc:convex:cliques} and~\eqref{sc:convex:indeps} as $F \in \ZZ^{k \times 2\cdot2^k}$ (notice that we also add the empty columns for the $z_i$ variables), and the matrix encoding the constraints~\eqref{sc:graver:sep} by $L \in \ZZ^{2^k \times 2\cdot2^k}$; thus $A = \left(\begin{smallmatrix}F \\ L\end{smallmatrix}\right)$.

\sv{\begin{lemma}[\appmark]\label{lem:tw}}
\lv{\begin{lemma}\label{lem:tw}}
In Model~\ref{mod:sc:graver} it holds that $\tw_D(A) \leq k+1$.
\end{lemma}
\sv{\begin{proof}[Proof Idea]
$G_D(F)$ is a $k$-clique $K_k$, and $G_D(L)$ is a $2^k$-path $P_{2^k}$.
Thus, $G_D(A)$ are these two graphs connected by all possible edges, and we construct a path decomposition, whose consecutive nodes contain $G_D(F)$ and consecutive vertices of $G_D(L)$.
\end{proof}}
\toappendix{
\sv{\begin{proof}[Proof of Lemma~\ref{lem:tw}]}
\lv{\begin{proof}}
We shall construct a tree decomposition of $G_D(A)$ of width $k+2$.
The decomposition is a path and has $|\Gamma| - 1$ nodes, one for each except the largest $i \in \Gamma$, in increasing order.
We put all $k$ rows of $F$ in the bag of every node.
In addition to that the bag of the $i$-th node contains the $i$-th and $(i+1)$-st row of $L$.
It is not difficult now to check that this indeed satisfies the definition of a tree decomposition.
\end{proof}}

\sv{\begin{lemma}[\appmark]\label{lem:graver}}
\lv{\begin{lemma}\label{lem:graver}}
In Model~\ref{mod:sc:graver} it holds that $g_1(A) \leq k^{\OhOp{k}}$.
\end{lemma}
\sv{\begin{proof}[Proof Idea] We first simplify the structure of $L$ by deleting duplicitous columns, and then explicitely construct a decomposition of any $\veh$ s.t. $L \veh = \mathbf{0}$ into conformal vectors $\veg$ of small $\ell_1$-norm. Combining with known bounds on matrices with few rows ($F$) and stacked matrices ($A$) yields the bound.\end{proof}}
\toappendix{
The idea behind the proof is as follows.
Since $A = \left(\begin{smallmatrix}F \\ L\end{smallmatrix}\right)$ is a matrix obtained by stacking the two blocks $F$ and $L$, the bound on $g_1(A)$, the largest coefficient in an element of the Graver basis of $A$, can be estimated using the following lemma for stacked matrices.
\begin{lemma}[Stacking lemma~{\cite[Lemma 3.7.6]{DeLoeraEtAl2013}}]
$g_1\left(\left(\begin{smallmatrix}F \\ L\end{smallmatrix}\right)\right) = g_1(F \cdot \G(L)) \cdot g_1(L)$
\end{lemma}
Here, $\G(L)$ is a matrix whose columns are vectors from the Graver basis of $L$.
Thus, we need to determine $g_1(L)$ and $g_1(F \cdot \G(L))$.
For the first bound we provide the following technical lemma.
\begin{lemma}\label{lem:graverBoundForLowerMatrix}
$g_1(L) \leq | \Gamma | + 1$.
Moreover, for every vector $\left(\begin{smallmatrix} \veg^z \\ \veg^x \end{smallmatrix}\right) \in \G(L)$ we have $\| \veg^x \|_1 \le 2$.
\end{lemma}
The rest is a quite straightforward application of the stacking lemma.
\begin{proof}[Proof of Lemma~\ref{lem:graver}]
Consider the matrix $F \cdot \G(L)$: it is a matrix with $k$ rows with entries bounded by the maximum of $\vef^{\intercal} \veg$ taking $\vef$ to be a row of $F$ and $\veg \in \G(L)$.
Trivially, $\|\vef\|_\infty \leq 1$ and Lemma~\ref{lem:graverBoundForLowerMatrix} yields that $\| \veg \|_1 \le |\Gamma| \le 2^k$, so we have $\|F \cdot \G(L)\|_\infty \leq 2^k$.
However, if we split $\vef$ naturally into two parts correspondding to the two groups of variables $\vef = \left(\begin{smallmatrix} \vef^z \\ \vef^x \end{smallmatrix}\right)$, we observe that $\vef^z = \mathbf{0}$ for every row $\vef$ of $F$.
By taking this and the latter part of Lemma~\ref{lem:graverBoundForLowerMatrix} into account, we arrive at $\|F \cdot \G(L)\|_\infty \leq 2$.

Eisenbrand et al.~\cite[Lemma 2]{EisenbrandHK:2018} show that, for a matrix $E \in \ZZ^{m \times N}$, a bound of $g_1(E) \leq (2m\|E\|_\infty+1)^m$ holds.
Plugging in, we obtain $g_1(F \cdot \G(L)) \leq (2k \cdot 2 + 1)^k = \OhOp{k^{k}}$, and using the stacking lemma, $g_1(A) \leq \OhOp{k^k} \cdot 2^{k} = k^{\OhOp{k}}$.
\end{proof}

\begin{proof}[Proof of Lemma~\ref{lem:graverBoundForLowerMatrix}]
We first simplify the structure of $L$.
It is known~\cite[Lemma 3.7.2]{DeLoeraEtAl2013} that repeating columns of a matrix $B$ does not increase $g_1(B)$; thus, it is enough to bound $g_1(L')$, where $L'$ is obtained from $L$ by deleting duplicitous columns.
Note that the columns corresponding to variables $x_I, x_{I'}$ are duplicitous whenever $\sigma(I) = \sigma(I')$.
So we may assume that $L'$ has the following form, obtaiend by keeping only one column for $\vex$ for every $i \in \Gamma$:
\begin{align}
\alpha_1 &= \beta_1 & \label{eqn:graver-base} \\
\alpha_i &= \alpha_{i-1} + \beta_{i} & \forall i \in [2,K], \label{eqn:graver}
\end{align}
for $K = |\Gamma|$.

First we are going to show that any integer vector $\veh$ with $L'\veh = \mathbf{0}$ can be written as a sum of integer vectors $\veg^1, \dots, \veg^M$ for some $M \in \N$, which satisfy $L' \veg^i = \mathbf{0}$, $\veg^i \sqsubseteq \veh$, and $\|\veg^i\|_1 \leq K + 1$, for all $i \in [M]$.
This is sufficient because while the $\veg^i$'s might not be elements of $\G(L')$ themselves, their maximum $\ell_1$-norm upper bounds $g_1(L')$.
To see this, observe that each such vector can be decomposed further into a $\sqsubseteq$-sum of vectors from a Graver basis of $L'$ and notice further that if $\veg' \sqsubseteq \veg$, then $\| \veg' \|_1 \le \| \veg \|_1$.

The rest of the proof is by induction on $\| \veh \|_1$.
If $\| \veh \|_1 = 0$, the claim clearly follows.
Otherwise let $\veh = \left(\begin{smallmatrix} \veh^\alpha \\ \veh^\beta \end{smallmatrix}\right)$ with $\| \veh \|_1 > 0$ and $L'\veh = \mathbf{0}$.
We have to find a nonzero vector $\veg = \left(\begin{smallmatrix} \veg^\alpha \\ \veg^\beta \end{smallmatrix}\right)$ with $\| \veg \|_1 \le K + 1$ and $L'\veg = \mathbf{0}$ such that $\veg \sqsubseteq \veh$ and $\| \veh - \veg \|_1 < \| \veh \|_1$.

To see this, first observe that if $\veh \neq \mathbf{0}$, then $\veh^\alpha \neq \mathbf{0}$.
Let $i \in [K]$ be such that $h^\alpha_1 = \cdots = h^\alpha_{i-1} = 0$ and $h^\alpha_i \neq 0$.
Now, using \eqref{eqn:graver-base} and \eqref{eqn:graver}, we observe the following.
\begin{claim}\label{clm:simpleSign-h-i-properties}
We have $h^\beta_1 = \cdots = h^\beta_{i - 1} = 0$ and $\sign(h^\alpha_i) = \sign(h^\beta_i)$.
\end{claim}
\begin{proof}
Since $h^\alpha_1 = \cdots = h^\alpha_{i-1}$, we have $h^\beta_1 = \cdots = h^\beta_{i-1}$.
Now \eqref{eqn:graver} together with $h^\alpha_{i-1} = 0$ results in $h^\alpha_i = h^\beta_i$ and the claim follows.
\cqed\end{proof}
Now there are two cases: either $\sign(h^\beta_{i + 1}) = -\sign(h^\alpha_i)$ or $\sign(h^\beta_{i+1}) \in \left\{ \sign(h^\alpha_i), 0 \right\}$.

Suppose $\sign(h^\beta_{i + 1}) = -\sign(h^\alpha_i)$.
Let $\veg^\alpha = \sign(h^\alpha_i) \cdot \vece_i$ and let $\veg^\beta = \sign(h^\alpha_i) \cdot (\vece_i - \vece_{i+1})$, where $\vece_i$ is the $i$-th unit vector, i.e., a vector with zeros everywhere except of the $i$-th coordinate, which is 1.
Observe that now $\veg^\alpha$ affects solely variable $\alpha_i$ and thus we have to care for the only two conditions containing $\alpha_i$ (recall $\alpha_{i-1} = 0$):
\[
\alpha_i = \beta_i \qquad \textrm{ and }\qquad \alpha_{i+1} = \alpha_i + \beta_{i+1} \,.
\]
This leaves us with a matrix with columns corresponding to $\alpha_i, \alpha_{i+1},\beta_i$, and $\beta_{i+1}$
\[\begin{pmatrix}
1 & 0 & -1 & 0 \\
-1 & 1 & 0 & -1
\end{pmatrix} \,.\]
The vector $\veg$,defined above, now corresponds to a vector $(1, 0, 1, 1)^\transpose$.
It is easy to see that this vector is in kernel of the matrix and, since $\alpha_i$ is the only affected $\alpha$-variable, we get $L'\veg = \mathbf{0}$ and we are done in this case.
Notice that in this case we have $\| \veg \|_1 = 3$.

Now suppose $\sign(h^\beta_{i+1}) \in \left\{ \sign(h^\alpha_i), 0 \right\}$.
We observe that this affects sign of $\alpha_{i+1}$.
\begin{claim}\label{clm:sameSignImplication}
If $\sign(h^\beta_{i+1}) \in \left\{ \sign(h^\alpha_i), 0 \right\}$, then $\sign(h^\alpha_i) = \sign(h^\alpha_{i+1})$.
\end{claim}
\begin{proof}
Suppose $\sign(h^\alpha_i) = 1$, the other case follows by a symmetric argument.
Then $h^\beta_{i+1}$ is nonnegative and by \eqref{eqn:graver} we obtain that $h^\alpha_{i+1}$ is a sum of a positive and a nonnegative number, thus a positive number as claimed.
\cqed\end{proof}
We are about to design a vector $\veg$ for which
\[
\alpha_{i+1} = \alpha_i + \beta_{i+1}
\]
holds.
Since $\sign(h^\beta_{i+1}) \in \left\{ \sign(h^\alpha_i), 0 \right\}$, we cannot use $g^\beta_{i+1}$ to fulfill the above condition and thus if $g^\alpha_i \neq 0$, then $\sign(g^\alpha_{i+1}) = \sign(g^\alpha_i)$.
Now if we set $g^\alpha_i = g^\beta_i = g^\alpha_{i+1} = \sign(g^\alpha_i)$ we fulfill all conditions \eqref{eqn:graver} (recall $\alpha_{i-1} = 0$).
But now the condition
\[
\alpha_{i+2} = \alpha_{i+1} + \beta_{i+2}
\]
is not satisfied.
However, we have essentially carried the difficulty from $\alpha_i$ to $\alpha_{i+1}$.
Since now either $\sign(h^\beta_{i + 1}) = -\sign(h^\alpha_i)$ or $\sign(h^\beta_{i+1}) \in \left\{ \sign(h^\alpha_i), 0 \right\}$, we arrive at the following.
\begin{claim}
Either
\begin{enumerate}
\item
there exists $j$ with $i < j \le K$ such that $\sign(h^\beta_i), \ldots, \sign(h^\beta_{j-1}) \in \left\{ \sign(h^\alpha_i), 0 \right\}$ and $\sign(h^\beta_j) = - \sign(h^\alpha_i)$ or
\item
it holds that $\sign(h^\beta_i), \ldots, \sign(h^\beta_{K}) \in \left\{ \sign(h^\alpha_i), 0 \right\}$.
\end{enumerate}
Let $j = K$ in the second case then we have $\sign(h^\alpha_i) = \cdots = \sign(h^\alpha_j)$.
\end{claim}
\begin{proof}
By repeated applications of Claim~\ref{clm:sameSignImplication} we get $\sign(h^\alpha_k) = \sign(h^\beta_{k+1})$ for all $i \le k \le j - 1$.
Initially the premise of Claim~\ref{clm:sameSignImplication} is what we suppose for this case and each application yields the premise of Claim~\ref{clm:sameSignImplication} for the next application.
\cqed\end{proof}
Let $j$ be defined as in the above claim.
Now, we finish the construction of $\veg$ by setting $\veg^\alpha = \sum_{k = i}^j \vece_k$.
In the first case of the above claim we let $\veg^\beta = \vece_i + \vece_j$ while in the second we have $\veg^\beta = \vece_i$.
It is not hard to verify that $L'\veg = \mathbf{0}$.
Indeed in the first case at index $j$ we essentially arrive to the situation described above when we argued about $\sign(h^\beta_{i + 1}) = -\sign(h^\alpha_i)$.
While if $j = K$, there is no carry, as there are no further rows of~$L'$.
The claimed bound on $g_1(L)$ follows by observing that we have $\| \veg \|_1 \le K + 1$ in both of the just described cases.
As for the latter part of the Lemma, observe that in every case we have $\left\| \veg^\beta \right\|_1 \le 2$ and notice that these variables correspond to the $x_I$ variables of the given model.
\end{proof}

} 


\bibliography{bib/minSumColoring,bib/fptlp}

\newpage\appendix\sloppy
\appendix

\appendixText

\section{Convex Integer Programming and Parameterized Complexity}
\label{sec:cip_pc}
In this section we overview existing results regarding minimization of convex (Subsection \ref{s:cimi_fd}), concave (Subsection~\ref{s:cima_fd}) and indefinite (Subsection~\ref{s:cimind_fd}) objectives in small dimension, and them move on to the rapidly growing area of IP in variable dimension (Subsection \ref{s:ilp_vd}).
The outline is inspired by Chapter 15 of the book \emph{50 Years of Integer Programming}~\cite{HemmeckeKLW10}, omitting some parts but including many recent developments.

\subsection{Convex Integer Minimization in Small Dimension} \label{s:cimi_fd}
Lenstra's result from 1983 shows that solving integer linear programming~\eqref{ILP} is polynomial when the integer dimension is small~\cite{Lenstra83}.
His result extends to the case where there are few integer variables but polynomially many continuous variables, called \emph{mixed ILP}:
\begin{equation}
\min \{\vew \vex \mid A \vex \leq \veb, \, \vex \in \ZZ^{n} \times \RR^{n'}\} \enspace . \tag{MILP} \label{MILP}
\end{equation}
Lenstra's algorithm was subsequently improved by Kannan~\cite{Kannan87} and Frank and Tardos~\cite{FrankTardos87} in two ways.
First, the required space was reduced from exponential to polynomial in the dimension, and second, running time dependency on the dimension $n$ was reduced from $2^{2^{\OhOp{n}}}$ to $n^{\OhOp{n}}$.
The main procedure in all of these algorithms is deciding \emph{feasibility}, i.e., is $\{\vex \mid A\vex \leq \veb\} \cap (\ZZ^n \times \RR^{n'})$ nonempty?
In order to optimize one does binary search over the objective, as described by Fellows et al.~\cite{FellowsLMRS08}.
We would like to point out that while Lenstra's result is old, we are aware of only a few~\cite{BredereckFNST:2015,JansenS:2010} applications which involve \emph{mixed} ILPs.

\begin{theorem}[Frank and Tardos~\cite{FrankTardos87}, Fellows et al.~\cite{FellowsLMRS08}] \label{thm:milp_fpt}
It is possible to solve~\eqref{MILP} using \mbox{$\OhOp{n^{2.5n} \cdot \poly(n') \cdot \langle A, \veb, \vew \r}$} arithmetic operations and space polynomial in $(n+n') \cdot \langle A, \veb, \vew \r$.
\end{theorem}

This result was later generalized to minimizing a quasiconvex function over a convex set, i.e., problem~\eqref{IP} with $f$ quasiconvex.
A function $f\colon \RR^n \rightarrow \RR$ is called \emph{quasiconvex} if for every $\alpha \in \RR$, the \emph{level set} $\{\vex \in \RR^n \mid f(\vex) \leq \alpha\}$ is a convex subset of $\RR^n$.
The first to show this was Grötschel, Lovász and Schrijver in their famous book~\cite[Theorem 6.7.10]{GLS}.
Unlike above, all of the following results require space exponential in the dimension.
Also, none of the cited results explicitely deals with the mixed integer case; however it is folklore that this is \FPT as well.

The subsequent research diverged in several directions.
The main difference between the papers we discuss is in the assumptions on the representation of the convex set $S$.
Since there is, strictly speaking, no ``better'' or ``worse'' assumption, choosing one is a matter of preference with respect to the specific scenario.
Another difference is in the motivation: some authors seek to achieve better time complexity while others contribute by simplifying existing proofs.
Our list is categorized according to the assumptions on the representation of $S$.

\subparagraph*{Semialgebraic convex set.}
Khachiyan and Porkolab~\cite{KhachiyanP00} state their result for minimizing a quasiconvex function over a \emph{semialgebraic convex set}; without going into technical details, let us say that these are closely related to spectrahedra, the solution spaces of semidefinite programs.
Independently, convex sets and semialgebraic sets have been studied for a long time, but together they have been studied only in the past ten years as \emph{Convex Algebraic Geometry}; cf. a book on the topic by Blekherman, Parillo and Thomas~\cite{BlekhermanPT12}.
A drawback of this result is an exponential dependence on the number of polynomials defining the semialgebraic convex set.

\begin{theorem}[Khachiyan and Porkolab~\cite{KhachiyanP00}] \label{thm:semialgebraic_ip_fpt}
Problem~\eqref{IP} with $f$ quasiconvex and $S$ a semialgebraic convex set defined by $k$ polynomials is \FPT with respect to $k$ and $n$.
\end{theorem}

\subparagraph*{Quasiconvex polynomials.}
Heinz~\cite{Heinz05} studied a more specific case of minimizing a quasiconvex polynomial over a convex set given by a system of quasiconvex polynomials, that is, polynomials that are quasiconvex functions.
His result improves over Khachiyan and Porkolab in terms of time complexity, dropping the exponential dependence on the number of polynomials.
The dependence on the dimension $n$ is $\OhOp{2^{n^3}}$, which was further improved by Hildebrand and Köppe~\cite{HildebrandK13} to $n^{\OhOp{n}}$.
The latter result can be stated as follows.
Let $\hat{F},F_1, \dots, F_m \in \ZZ[\vex] = \ZZ[x_1, \dots, x_n]$ be polynomials with integer coefficients. Then we get the \textsc{Convex Polynomial IP} problem
\begin{equation}
\min \{\hat{F}(\vex) \mid F_i(\vex) < 0\, \forall i \in [m], \, \vex \in \ZZ^n\} \enspace .\tag{CPIP}\label{CPIP}
\end{equation}

\begin{theorem}[Hildebrand and Köppe~\cite{HildebrandK13}] \label{thm:polynomial_ip_fpt}
Given a~\eqref{CPIP} instance with \mbox{$\mathcal{F} = \left\{ \hat{F}, F_1, \ldots, F_m \right\}$}.
Let ${d \geq 2}$ be the upper bound on the degree of each \mbox{$F \in \mathcal{F}$}, $M$ is the maximum number of monomials in each \mbox{$F \in \mathcal{F}$} and $\ell$ bounds the binary length of the coefficients of $F$.
Then it can be solved in time:
  \begin{itemize}
  \item $n^{\OhOp{n}} \cdot m(r\ell Md)^{\OhOp{1}}$, thus \FPT with respect to the dimension $n$, if the feasible region is bounded such that $r$ is the binary encoding length of that bound with \mbox{$r\leq \ell d^{\OhOp{p}}$},
  \item $d^{\OhOp{n}} n^{2n} \cdot m\ell^{\OhOp{1}}$, thus \FPT with respect to the dimension $n$ and the maximum degree $d$, if the feasible region is unbounded.
  \end{itemize}
\end{theorem}
Note that, in particular, the running time is polynomial with respect to the number of polynomials~$m$.
We also note that the quantities $r$ and $\ell$ are natural and effectively appear in the $\langle A, \veb, \vew\r$ runtime term of Theorem~\ref{thm:milp_fpt}.

An advantage of representing $S$ and $f$ by polynomials is that the representation is ``explicit'', in contrast to representing them by an oracle.
Polynomial objectives appear for example in scheduling~\cite{KnopK:2016,MnichW:2015} where models of small dimension represent jobs by multiplicities, and an objective such as $\sum w_j C_j$ (sum of weighted completion times) becomes quadratic in this encoding.
The drawback of representing $S$ or $f$ by polynomials is that it is sometimes unnatural, with piece-wise linear convex constraints (Model~\ref{mod:cds:convex}) or objectives~\cite{GajarskyHKO:2017}.

\subparagraph*{Oracles.}
Further research lead to splitting convex IP (i.e., problem~\eqref{IP} with $f$ convex) in two independent parts to allow more focus on each of them.
The first part is showing that a certain problem formulation (such as quasiconvex polynomial inequalities, semialgebraic set etc.) can be used to give a set of geometric oracles.
The second part is to show that, given these oracles, solving a convex IP can be done in a certain time.

This approach is taken by Dadush, Peikert and Vempala~\cite{DadushPV11} who further improve the time complexity of Hildebrand and Köppe~\cite{HildebrandK13} when the convex set is given by three oracles: a so-called \emph{weak membership}, \emph{strong separation} and \emph{weak distance} oracles.
Observe that the running time of Theorem~\ref{thm:polynomial_ip_fpt} can be rewritten as $\OhStarOp{n^{2n}}$; Dadush et al. improve it to $\OhStarOp{n^{\frac{4}{3}n}}$.
Moreover, Dadush claims in his PhD thesis~\cite{Dadush12} a randomized $\OhStarOp{n^n}$ algorithm; for derandomization cf.~\cite{DadushV:2013}.
(Here and in the following we use the $\Oh^*$ notation which suppresses polynomial factors.)

This sequence of results can be seen as a part of a race for the best running time.
Dadush~\cite{Dadush12} classifies existing algorithms as \emph{Lenstra-type} and \emph{Kannan-type}, depending on the space decomposition they use (hyperplane and subspace, respectively).
The type of algorithm determines the best possible running time -- Lenstra-type algorithms depend on a so-called \emph{flatness theorem}, which gives a lower-bound $\OhStarOp{n^n}$.
The best known Lenstra-type algorithm is the $\OhStarOp{n^{\frac{4}{3}n}}$ algorithm of Dadush et al.~\cite{DadushPV11}.
Note that both Theorem~\ref{thm:semialgebraic_ip_fpt} and~\ref{thm:polynomial_ip_fpt} are Lenstra-type algorithms.
On the other hand, Kannan-type algorithms could run as fast as $\OhStarOp{({\log n})^{n}}$ if a certain conjecture of Kannan and Lovász holds~\cite[Theorem 7.1.3]{Dadush12}.
The~$\OhStarOp{n^n}$ algorithm given in Dadush's thesis~\cite{Dadush12} is Kannan-type.
It is also worth noting that the only known lower bound for convex IP in general is the trivial one of $\OhStarOp{2^n}$ (by encoding SAT as binary ILP).

The oracle approach is also taken by Oertel, Wagner and Weismantel~\cite{OertelWW14}.
They show that a convex IP given by a so-called \textit{first order evaluation oracle} can be reduced to several MILP subproblems, which are readily solved by existing solvers (implementing for example Theorem~\ref{thm:milp_fpt}).
In an earlier version of this paper~\cite{OertelWW12} the authors take a more generic approach requiring a set of oracles to solve a minimization problem, and discuss how to construct these oracles specifically for the~\eqref{CPIP} problem.

\subsection{Concave Integer Minimization in Small Dimension} \label{s:cima_fd}
When we make the step from a \emph{linear} to a general \emph{quasiconvex} objective function, we have to distinguish carefully between convex minimization and maximization, or equivalently, between minimizing a convex and a concave function.
Here we mention one result that can be applied in the concave minimization case.

\subparagraph*{Vertex Enumeration.}
Provided bounds on the encoding length and number of inequalities, there is a good bound on the number of vertices of the integer hull of a polyhedron:

\begin{theorem}[Cook et al.~\cite{CookHKM92}]
Let \mbox{$P = \left\{ \vex \in \RR^n \mid A\vex \leq \veb \right\}$} be a~rational polyhedron with \mbox{$A \in \QQ^{m \times n}$} and let $\phi$ be the largest binary encoding size of any of the rows of the system \mbox{$A\vex \leq \veb$}.
Let \mbox{$P^I = \textrm{conv}(P \cap \ZZ^n)$} be the integer hull of $P$.
Then the number of vertices of $P^I$ is at most \mbox{$2m^n (6n^2 \phi)^{n-1}$}.
\end{theorem}

Since Hartmann~\cite{Hartmann89} also gave an algorithm for enumerating all the vertices running in polynomial time in small dimension, it is possible to evaluate the concave objective function on each of them and pick the best.
The crucial observation which makes this sufficient is that any concave objective is minimized on the boundary, which will be a vertex.
Moreover, in parameterized complexity we are often dealing with combinatorial problems whose ILP description only contains numbers encoded in unary, implying that the \emph{encoding length} $\phi$ is logarithmic in the size of the instance $|I|$.
Since $(\log |I|)^k$ for fixed $k$ is order $o(|I|)$~\cite[Lemma 6.1]{ChitnisCHM:15}, convex integer maximization is \FPT in all such cases.

\subsection{Indefinite Optimization in Small Dimension} \label{s:cimind_fd}
Results regarding optimizing indefinite polynomials in fixed dimension are few, indicating this area merits much attention.
De Loera et al.~\cite{DeLoeraHKW:2008} show that optimizing an indefinite non-negative polynomial over the mixed-integer points in small dimensional polytopes admits a fully-polynomial time approximation scheme (FPTAS); however, the runtime of this algorithm is \XP from the perspective of parameterized complexity, and it has not yet found applications.

Hildebrand et al.~\cite{HildebrandWZ:2016} recently also provided an FPTAS, however, their results are incomparable to the previous one.
On one hand, their results are stronger because they use a different notion of approximation, and because they do not require the non-negativity of the objective function.
On the other hand, there are additional requirements on the polynomial, namely that it is quadratic and has at most one negative or at most one positive eigenvalue.

The most significant contribution from the perspective of parameterized complexity is an \FPT algorithm for \textsc{Quadratic Integer Programming} by Lokshtanov~\cite{Lokshtanov:2015}, independently also discovered by Zemmer~\cite{Zemmer:2017}:

\begin{theorem}[{Lokshtanov~\cite{Lokshtanov:2015}, Zemmer~\cite{Zemmer:2017}}]
Let $Q \in \ZZ^{n \times n}$ and $f(\vex) = \vex^{\intercal} Q \vex$.
Then problem~\eqref{LinIP} is \FPT parameterized by $n$, $\|A\|_\infty$, and $\|Q\|_\infty$.
\end{theorem}

While this parameterization may seem very restrictive, it lead to the resolution of a major open problem regading the parameterized complexity of \textsc{Minimum Linear Arrangement} parameterized by the vertex cover number.

\subsection{Integer Linear Programming in Variable Dimension} \label{s:ilp_vd}

Two major well-known cases of linear programs (LPs) that can be solved
integrally in polynomial time are LPs in small dimension (as discussed
above) and LPs given by totally unimodular matrices (such as flow
polytopes).
A large stream of research of the past 20 years has very recently converged on a result largely explaining the parameterized complexity of IP in terms of the structural complexity of the matrix $A$.
We are interested in the parameterizations of three graphs associated to the constraint matrix $A$:
\begin{enumerate}
\item The \emph{primal graph} $G_P(A)$, which has a vertex for every column, and two vertices share an edge if a row exists where both corresponding entries are non-zero.
\item The \emph{dual graph} $G_D(A) = G_P(A^{\intercal})$, which is the primal graph of the transpose of a matrix.
\item The \emph{incidence graph} $G_I(A)$, which has a vertex for every row and every column, and two vertices share an edge if they correspond to a row-column coordinate which is non-zero.
\end{enumerate}
Specifically, we are interested in the treedepth and treewidth of these graphs, yielding six parameters: primal/incidence/dual treedepth/treewidth, denoted $\td_P(A)$, $\tw_P(A)$, $\td_I(A)$, $\tw_I(A)$, $\td_D(A)$ and $\tw_D(A)$.
The fundamental result can be phrased as follows:

\begin{theorem}[{\cite[Theorems 5 and 6]{KouteckyLO:2018}}] \label{thm:graver}
There are computable functions $h_P$ and $h_D$ such that problem~\eqref{StandardLinIP} with $f$ a separable convex function can be solved in time:
\begin{itemize}
\item $h_P(\|A\|_\infty, \td_P(A)) n^3 \langle f_{\max}, \vel, \veu, \veb \r$, and
\item $h_D(\|A\|_\infty, \td_D(A)) n^3 \langle f_{\max}, \vel, \veu, \veb \r$.
\end{itemize}
\end{theorem}
In the case of ILP (linear objective), these results can even be made \emph{strongly polynomial}, i.e., not depending on the encoding lengths $\langle \vew, \vel, \veu, \veb \r$.
Let us discuss in more detail how these results are obtained.

\subparagraph*{Graver basis optimization.}
A key notion is that of \emph{iterative augmentation}.
Most readers will be familiar that the \textsc{Max Flow} problem can be solved by starting from a zero flow, and iteratively augmenting it with paths; when no augmenting path exists, the flow is optimal.
The notion of a \emph{Graver basis} (cf. Definition~\ref{def:graver}) lets us extend this approach to \eqref{StandardLinIP} as follows.
Starting from some initial feasible point $\vex_0 \in \ZZ^n$, there either exists a $\veg \in \G(A)$ such that $\vex_0+\veg$ is feasible (i.e., $\vel \leq \vex_0 + \veg \leq \veu$) and augmenting (i.e., $f(\vex_0 + \veg) < f(\vex_0)$), or $\vex_0$ is guaranteed to be optimal.
This is not yet enough to ensure quick convergence to an optimal point $\vex^*$, but always augmenting with a \emph{Graver-best} step $\veg$ then also guarantees this.
Thus the question becomes in which cases it is possible to efficiently compute such Graver-best steps.
This turns out to depend on the primal and dual treewidth and the norms of elements of $\G(A)$; recall that $g_\infty(A) = \max_{\veg \in \G(A)} \|\veg\|_\infty$ and analogously $g_1(A) = \max_{\veg \in \G(A)} \|\veg\|_1$.
\begin{lemma}[{Primal and dual lemma~\cite[(roughly) Lemmas 22 and 25]{KouteckyLO:2018}}] \label{lem:primaldual}
A Graver-best step can be found in time
\begin{itemize}
\item $g_\infty(A)^{\Oh(\tw_P(A))} \cdot (n+m)$, and,
\item $g_1(A)^{\Oh(\tw_D(A))} \cdot (n+m)$.
\end{itemize}
\end{lemma}
The proof of this lemma uses two dynamic programs; the first is well known and goes back to Freuder~\cite{Freuder:1990,JansenK15}, the second was only recently described by Ganian et al.~\cite{GanianOR:2017}.

\subparagraph*{Graver basis norms.}
The next obvious question is: what IPs satisfy the assumptions of Lemma~\ref{lem:primaldual}?
Hemmecke and Schultz~\cite{HemmeckeS:2001} show (though not in those terms) that \emph{2-stage stochastic} matrices have small $g_\infty(A)$, and it is not hard to see that they have small $\td_P(A) \leq \tw_P(A)$.
This result was later extended by Aschenbrenner and Hemmecke~\cite{AschenbrennerH:2007} to \emph{multi-stage stochastic} matrices, which are in turn generalized (and simultaneously generalize) matrices with small primal treedepth $\td_P(A)$, so we have that $g_\infty(A) \leq h(\|A\|_\infty, \td_P(A))$ for some computable function $h$.

Similarly, it was shown~\cite{Onn:2010} that \emph{$n$-fold} matrices have small $g_1(A)$ and they also have small $\td_D(A)$.
Those are generalized by \emph{tree-fold} matrices introduced by Chen and Marx~\cite{ChenM:2018} who generalize (and are generalized by) matrices with small dual treedepth $\td_D(A)$.

Theorem~\ref{thm:graver} (and its previous versions) has found use for example in parameterized scheduling~\cite{ChenM:2018,KnopK:2016}, computational social choice and stringology~\cite{KnopKM:2017b,KnopKM:2017}, and the design of efficient polynomial time approximation schemes (EPTASes)~\cite{JansenKMR:2018}.

\subparagraph*{Incidence treedepth.}
We note that the classification result of Theorem~\ref{thm:graver} cannot be improved in any direction: allowing unary-sized coefficients $\|A\|_\infty$ gives~\W{1}-hardness, and relaxing treedepth to treewidth leads to \NP-hardness~\cite{KouteckyLO:2018}.

The complexity of parameterizing by $\td_I(A)$ and $\|A\|_\infty$ is wide open.
The simplest stepping stone seems to be so-called 4-block $n$-fold programs, which combine the structure of 2-stage stochastic and $n$-fold matrices.
4-block $n$-fold IP is known to be \XP parameterized by the block dimensions~\cite{HemmeckeKLW10}, and \FPT membership is an important open problem.
Recently, Chen et al.~\cite{ChenXS:2018} gave some indication that the problem might in fact be~\Wh{1}.

\subparagraph*{ILP with few rows.}
Restricting our attention to a simpler case then the one handled by Theorem~\ref{thm:graver} leads us to considering ILPs with few rows.
Papadimitriou showed that ILP is \FPT parameterized by $\|A\|_\infty$ and $m$~\cite{Papadimitriou:1981}.
His algorithm was recently sped up by Eisenbrand and Weismantel~\cite{EisenbrandW:2018} and in the special case without upper bounds also by Jansen and Rohwedder~\cite{JansenR:2018}.
Many approximation algorithms (especially EPTASes) contain a subroutine using Lenstra's algorithm to solve a certain configuration IP.
Provided that this IP has small coefficients, this step can be exponentially sped up by applying one of the aforementioned algorithms.
A good example is the algorithm of Lampis for \textsc{Coloring} on graphs of bounded neighborhood diversity~\cite{Lampis12}, which can be improved from $2^{2^{k^{\Oh(1)}}} \log |G|$ to $k^{\Oh(k)} \log |G|$ simply by replacing Lenstra's algorithm.

\subparagraph*{Miscellaneous results.}
We highlight that we are not aware of any uses of multi-stage stochastic IP in parameterized complexity, and it would be interesting to see what kind of problems it can model.
Another interesting result from this area which has not yet found applications is due to Lee et al.~\cite{LeeORW:2012}.
It states that minimizing even certain non-convex objectives is polynomial-time solvable provided the objective falls in the so-called \emph{quadratic Graver cone}.

One way how to view the results based on Graver bases is via the parameter \emph{fracture number}: a graph has a small fracture number if there exists a small subset of vertices whose deletion decomposes the graph into (possibly many) small components; note that the treedepth is always at most the fracture number.
Dvořák et al.~\cite{DvorakEGKO:2017} show that ILP parameterized by the largest coefficient and the constraint or variable fracture number of the primal graph is \FPT.
In the case of constraint fracture number, one must delete small set of vertices corresponding only to constraints of the ILP at hand.
The variable fracture number is defined accordingly.
They provide an equivalent instance of either $n$-fold IP or 2-stage stochastic IP. These results are subsumed by Theorem~\ref{thm:graver}, but the parameter mixed fracture number (allowing the deletion of both rows and columns of $A$) is interesting because it is equivalent to 4-block $n$-fold and could be useful to understand its complexity.

Jansen and Kratsch~\cite{JansenK15} studied the kernelizability of ILP and show that instances with bounded domains and bounded primal treewidth are efficiently kernelizable.
Moreover, they introduce so-called \emph{$r$-boundaried} ILPs which generalize totally unimodular ILPs and ILPs of bounded treewidth, and they give an FPT result regarding $r$-boundaried ILPs.

Finally, Ganian et al.~\cite{GanianOR:2017} show that ILP parameterized by incidence treewidth and the largest constraint partial sum of any feasible solution is \FPT.
They also combine primal treewidth with Lenstra's algorithm to obtain a new structural parameter called \emph{torso-width}, and give an \FPT algorithm for this parameterization.

\end{document}